\documentclass[JHEP,12pt]{article}
\usepackage{amssymb,amsmath,amsthm}
\usepackage{braket}
\usepackage{jheppub}
\makeatletter
\def\@fpheader{\mbox{}}
\makeatother%\usepackage{showkeys}
\usepackage{rotating}
\usepackage{bbold}
\usepackage{subcaption}
\usepackage{comment} 
\newtheorem{thm}{Theorem}[section]
 
\newtheorem{lem}[thm]{Lemma}
\theoremstyle{remark}
\newtheorem{defn}[thm]{Definition}
\newtheorem{rem}[thm]{Remark}
\newtheorem{conj}[thm]{Conjecture}
\newtheorem{conv}[thm]{Convention}

\newcommand{\deltabar}{\delta\hspace*{-0.2em}\bar{}\hspace*{0.2em}}
\DeclareMathOperator{\tr }{tr}
\DeclareMathOperator{\cl }{cl}
\DeclareMathOperator{\setint }{int}
\DeclareMathOperator{\A}{Area}

\renewcommand{\S}{S_{\rm gen}}

\author{Raphael Bousso and Geoff Penington}

\affiliation{Center for Theoretical Physics and Department of Physics,\\
University of California, Berkeley, CA 94720, U.S.A. %and \\
} 

\emailAdd{bousso@berkeley.edu}
\emailAdd{geoffp@berkeley.edu}

\title{Entanglement Wedges For Gravitating Regions}
\abstract{Motivated by properties of tensor networks, we conjecture that an arbitrary gravitating region $a$ can be assigned a {\em generalized entanglement wedge} $E\supset a$, such that quasi-local operators in $E$ have a holographic representation in the full algebra generated by quasi-local operators in $a$. The universe need not be asymptotically flat or AdS, and $a$ need not be asymptotic or weakly gravitating. 

On a static Cauchy surface $\Sigma$, we propose that $E$ is the superset of $a$ that minimizes the generalized entropy. We prove that $E$ satisfies a no-cloning theorem and appropriate forms of strong subadditivity and nesting. If $a$ lies near a portion $A$ of the conformal boundary of AdS, our proposal reduces to the Quantum Minimal Surface prescription applied to $A$. We also discuss possible covariant extensions of this proposal, although none prove completely satisfactory.

Our results are consistent with the conjecture that information in $E$ that is spacelike to $a$ in the semiclassical description can nevertheless be recovered from $a$, by microscopic operators that break that description. We thus propose that $E$ quantifies the range of holographic encoding, an important nonlocal feature of quantum gravity, in general spacetimes.}

\begin{document}
\maketitle

\section{Introduction}

We seek a theory of quantum gravity that can describe our own universe. One possible approach to this challenge is to analyze the well-known theories~\cite{Maldacena:1997re} that instead describe asymptotically anti-de Sitter space (AdS).  In particular, one would like to identify elements of the AdS/CFT correspondence that have broader validity. The entanglement wedge prescription is a crucial part of the AdS/CFT dictionary: it identifies the spacetime region dual to a given CFT subregion. We will propose a more general prescription that associates an entanglement wedge to arbitrary regions in arbitrary spacetimes.

\subsection{Entanglement Wedge of a Boundary Region}

The Ryu-Takayanagi (RT) proposal~\cite{Ryu:2006bv} was first introduced in the context of the AdS/CFT correspondence, as a method for computing the von Neumann entropy of a stationary CFT state reduced to any spatial region $A$ of the boundary. The proposal was refined to allow for time-dependent states by Hubeny, Rangamani, and Takayanagi~\cite{Hubeny:2007xt}. Quantum corrections were first included by Faulkner, Lewkowycz, and Maldacena~\cite{Faulkner:2013ana}. Engelhardt and Wall~\cite{Engelhardt:2014gca} introduced the present, most powerful formulation of RT in terms of the Quantum Extremal Surface QES$_{\rm min}(A)$ anchored on $A$ with smallest generalized entropy among all such surfaces. This prescription can be applied in the semiclassical regime to any order in $G\hbar$. 

The {\em entanglement wedge}, $\mathrm{EW}(A)$, is defined as the homology wedge, \emph{i.e.}, the causal development of any spatial region bounded by $\mbox{QES}_{\rm min}(A)$ and $A$. There is strong evidence that in a large class of states,\footnote{See~\cite{Hayden:2018khn,Akers:2020pmf} for subtleties and appropriate refinements of the entanglement wedge. Our proposal admits analogous refinements, which may be important for its generalization to time-dependent settings; see Sec.~\ref{sec-covariantdiscussion}.} $\mathrm{EW}(A)$ is precisely the bulk dual to the boundary region $A$~\cite{Wall:2012uf,Jafferis:2015del}.\footnote{See \cite{Bousso:2012sj,Hubeny:2012wa,Czech:2012bh} for early work on this problem.} That is, any local bulk operator in $\mathrm{EW}(A)$ can be implemented, or ``reconstructed,'' by a CFT operator in the algebra associated to the region $A$; and no bulk operator outside $\mathrm{EW}(A)$ can be so implemented. This result is known as entanglement wedge reconstruction.
\begin{figure}[t]
\begin{center}
  \includegraphics[scale=.4]{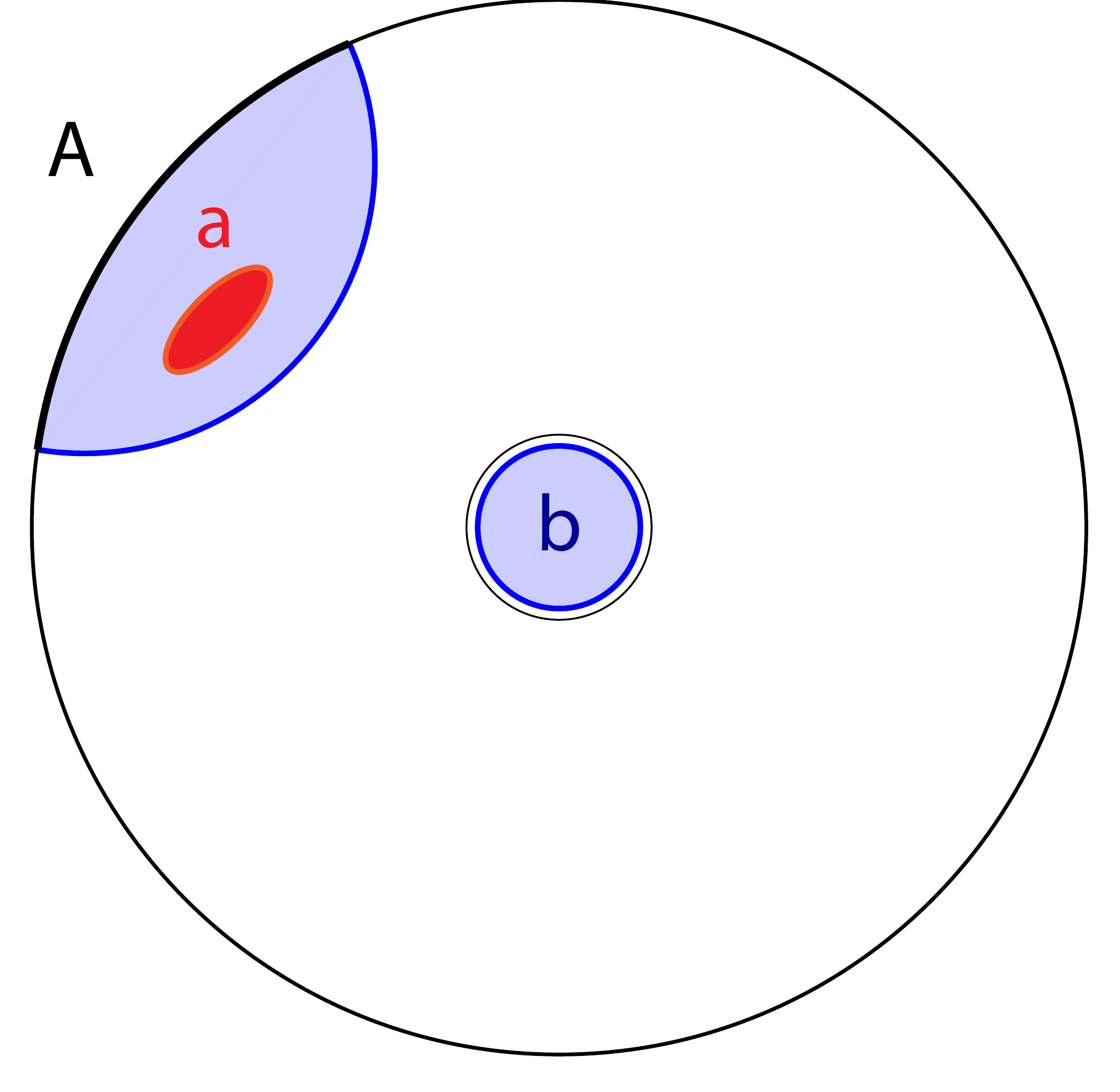} 
\end{center}
\caption{Entanglement island. Hawking radiation emitted from a black hole (grey circle) is collected in the reservoir $a$ (red) located near the boundary region $A$. Before the Page time, the (ordinary) entanglement wedge $EW(A)$ consists only of the blue region adjacent to $A$. After the Page time, $EW(A)$ also contains an island $b$ in the black hole interior.}
\label{fig-island}
\end{figure}

Its definition does not require the entanglement wedge to be connected. Any component of $\mathrm{EW}(A)$ that is not connected to $A$ is called an {\em entanglement island}~\cite{Penington:2019npb,Almheiri:2019psf}. An island arises naturally in the course of black hole evaporation, if the radiation is stored in a distant, localized reservoir $a$~\cite{Bousso:2019ykv}. Consider a small boundary region $A$ whose entanglement wedge contains $a$, as shown in Fig.~\ref{fig-island}. After the Page time, \emph{i.e.}, when the entropy of the radiation exceeds the Bekenstein-Hawking entropy of the remaining black hole, $\mathrm{EW}(A)$ will also include an island $b$ just inside of the black hole horizon. The inclusion of $b$ in $\mathrm{EW}(A)$ decreases the generalized entropy since the radiation entropy in $a$ is purified by the Hawking partners in $b$, and the price paid is only the Bekenstein-Hawking entropy of the island boundary (roughly that of the black hole).

\subsection{Going Beyond AdS}
\label{sec-beyond}

With certain assumptions, the entanglement wedge prescription can be derived from the gravitational path integral~\cite{Lewkowycz:2013nqa}. This implies that the entanglement wedge prescription can be used to compute the entropy not only of subregions of the conformal boundary of AdS, but of more general systems that are coupled to a gravitating spacetime, regardless of whether an exact description is known.\footnote{The earliest example of this phenomenon was the observation that, if a nongravitating external system $\mathsf{R}$ purifies matter near the QES that bounds $\mathrm{EW}(A)$, then access to $\mathsf{R}$ can increase the size of the reconstructible bulk region \cite{Hayden:2018khn, Akers:2019wxj}. In other words, $\mathrm{EW}(A\cup \mathsf{R})$ can differ from $\mathrm{EW}(A)$.}

Crucially, this suggests that the entanglement wedge prescription is not confined to the context of AdS/CFT. Here we propose that an entanglement wedge prescription should be applicable in general spacetimes. Like the Generalized Second Law~\cite{Bekenstein:1972tm,Wall:2011hj} and the Covariant Entropy Bound~\cite{Bousso:1999xy,Bousso:2015mna}, the entanglement wedge reveals aspects of the (usually unknown) full quantum gravity theory, extracted from the gravitational path integral.\footnote{When a full theory \emph{is} available, it can be used to verify the predictions of the entanglement wedge prescription. For example, by the AdS/CFT duality, the Page curve for an evaporating AdS black hole follows independently from the scrambling properties and the unitarity of the CFT.}

The full generality of the entanglement wedge prescription --- its applicability beyond the AdS/CFT correspondence --- was long obscured by the lack of nontrivial bound\-ary-anchored extremal surfaces in non-AdS spacetimes or auxiliary systems. Entanglement islands, however, are detached from the conformal boundary by definition. Hence, islands furnish nontrivial examples of entanglement wedges of systems that have nothing to do with the conformal boundary of AdS.
\begin{figure}[t]
\begin{subfigure}{.48\textwidth}
  \centering
 \includegraphics[width = 0.97\linewidth]{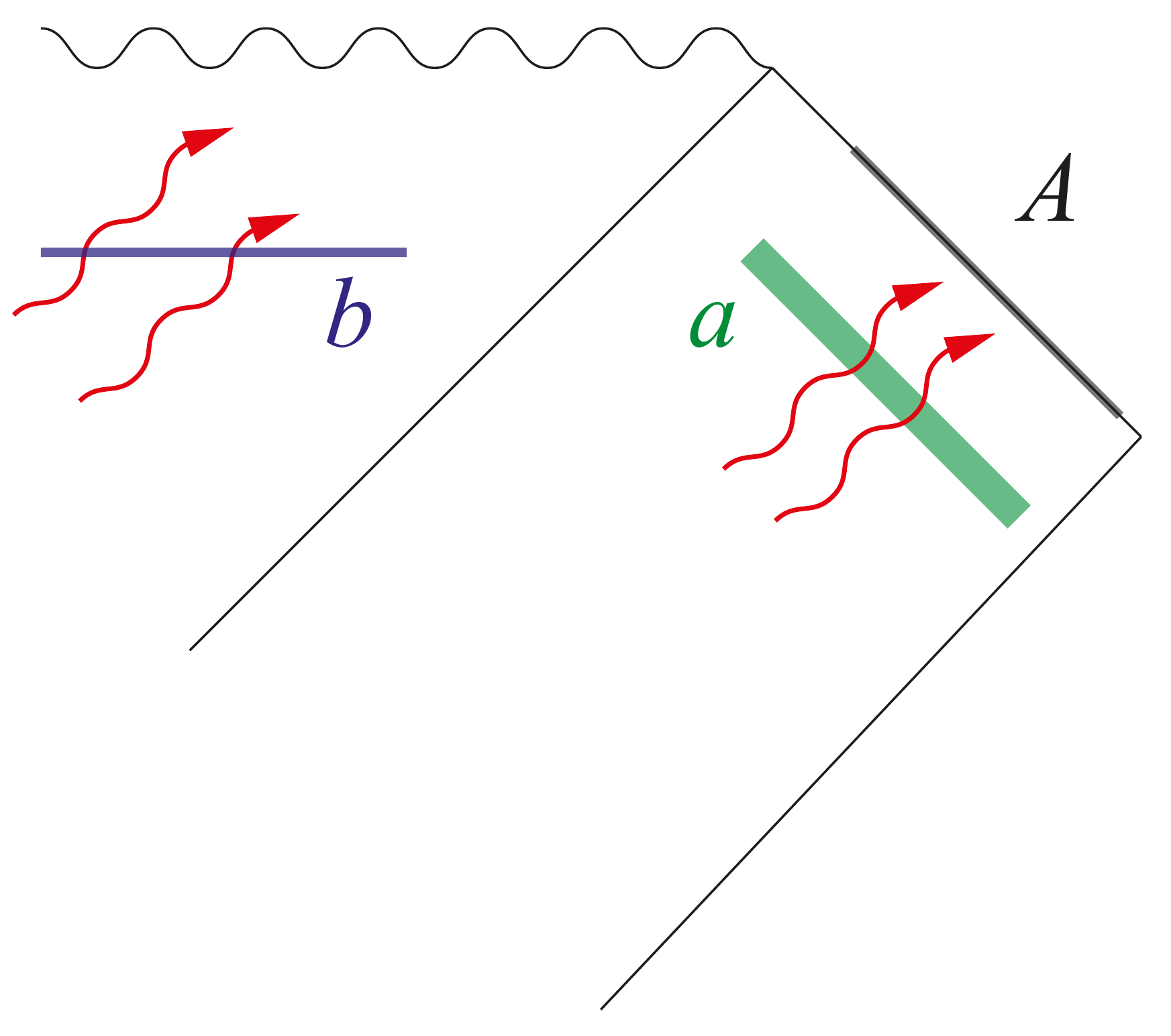}
\end{subfigure}
\begin{subfigure}{.48\textwidth}
  \centering
 \includegraphics[width = 0.95\linewidth]{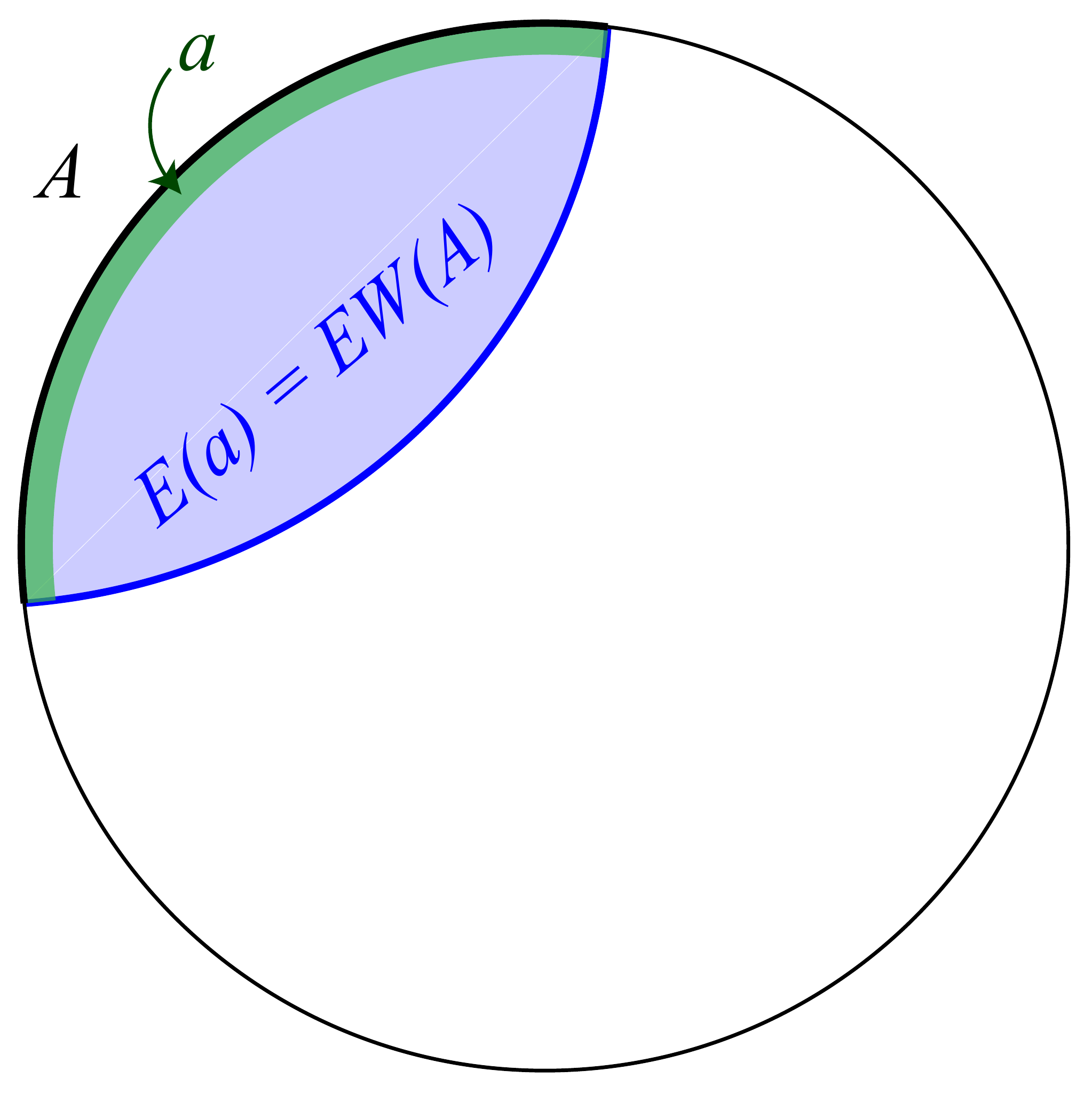}
\end{subfigure}
\caption{Examples that motivate and constrain our definition of a \emph{generalized entanglement wedge}, a map $E: a\to E(a)$, that takes bulk regions as its input. \emph{Left:} Hawking radiation arrives on a portion $A$ of future null infinity. After the Page time, the (original) entanglement wedge of $A$ includes an island: $\mathrm{EW}(A)=A\cup b$. When the radiation is still in the distant spacetime region $a$ (green shaded), it should already possess a generalized entanglement wedge, such that $E(a)\supset a\cup b$.  \emph{Right:} the (original) entanglement wedge $\mathrm{EW}(A)$ of an AdS boundary region $A$  can be regarded as the bulk algebra of operators encoded in the CFT algebra of operators in $A$. The latter is generated by local operators in $A$, which are dual to quasi-local bulk operators in the near-boundary region $a$. Hence $a$ should possess a generalized entanglement wedge such that $E(a)=\mathrm{EW}(A)$.}
\label{fig-motivate}
\end{figure}

For example, if the Hawking radiation is transferred to an external system $\mathsf{R}$, then after the Page time, $\mathsf{R}$ will possess a nontrivial entanglement wedge: $\mathrm{EW}(\mathsf{R}) = \mathsf{R} \cup b$. (Islands were first discovered in this setting.) This should also be the case for Hawking radiation far from the black hole in asymptotically flat spacetimes (see Fig.~\ref{fig-motivate}, left panel). In both cases, the presence of an entanglement island $b$ is crucial to deriving the Page curve and thus the unitarity of the Hawking process. And in neither case is the entanglement wedge computed for a portion of the conformal boundary of AdS.

What should be regarded as the proper input for computing an entanglement wedge: a portion of the conformal boundary? Nongravitating systems outside the spacetime? In this paper, we will propose that the answer is neither: rather, the entanglement wedge prescription should be formulated as a map that takes any gravitating spacetime region as input, and outputs another (equal or larger) bulk spacetime region. The entanglement wedges of nongravitating systems and of boundary regions should be interpreted as limits of this more general prescription that arise if the input region is weakly gravitating or approaches the asymptotic boundary.

Because the real world has gravity, such a proposal might be testable at least in principle, unlike a boundary-to-bulk map. Importantly, it should extend the notion of an entanglement wedge to arbitrary spacetimes, including cosmology. 

For static spacetimes, we propose that the generalized entanglement wedge $E(a)$ is simply the superset of $a$ with smallest generalized entropy.\footnote{Throughout, the symbols $\subset$ and $\supset$ are understood to allow set equality.} We motivate this prescription using tensor network toy models of quantum gravity and show that it satisfies a number of desirable properties that suggest it correctly quantifies the range of holographic encoding.

\begin{figure}[t]
\begin{subfigure}{.48\textwidth}
  \centering
 \includegraphics[width = 0.6\linewidth]{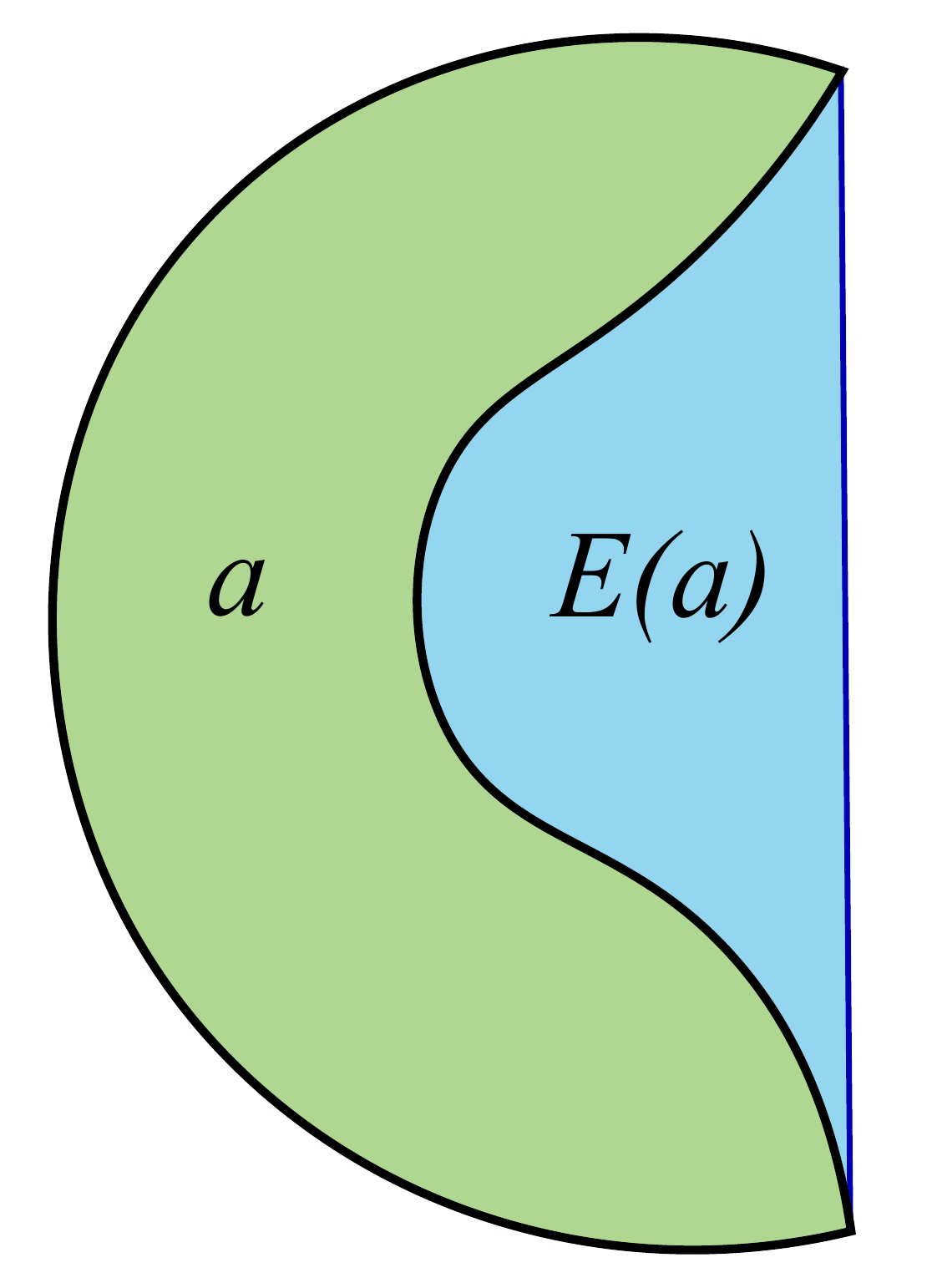}
\end{subfigure}
\begin{subfigure}{.48\textwidth}
  \centering
 \includegraphics[width = 0.75\linewidth]{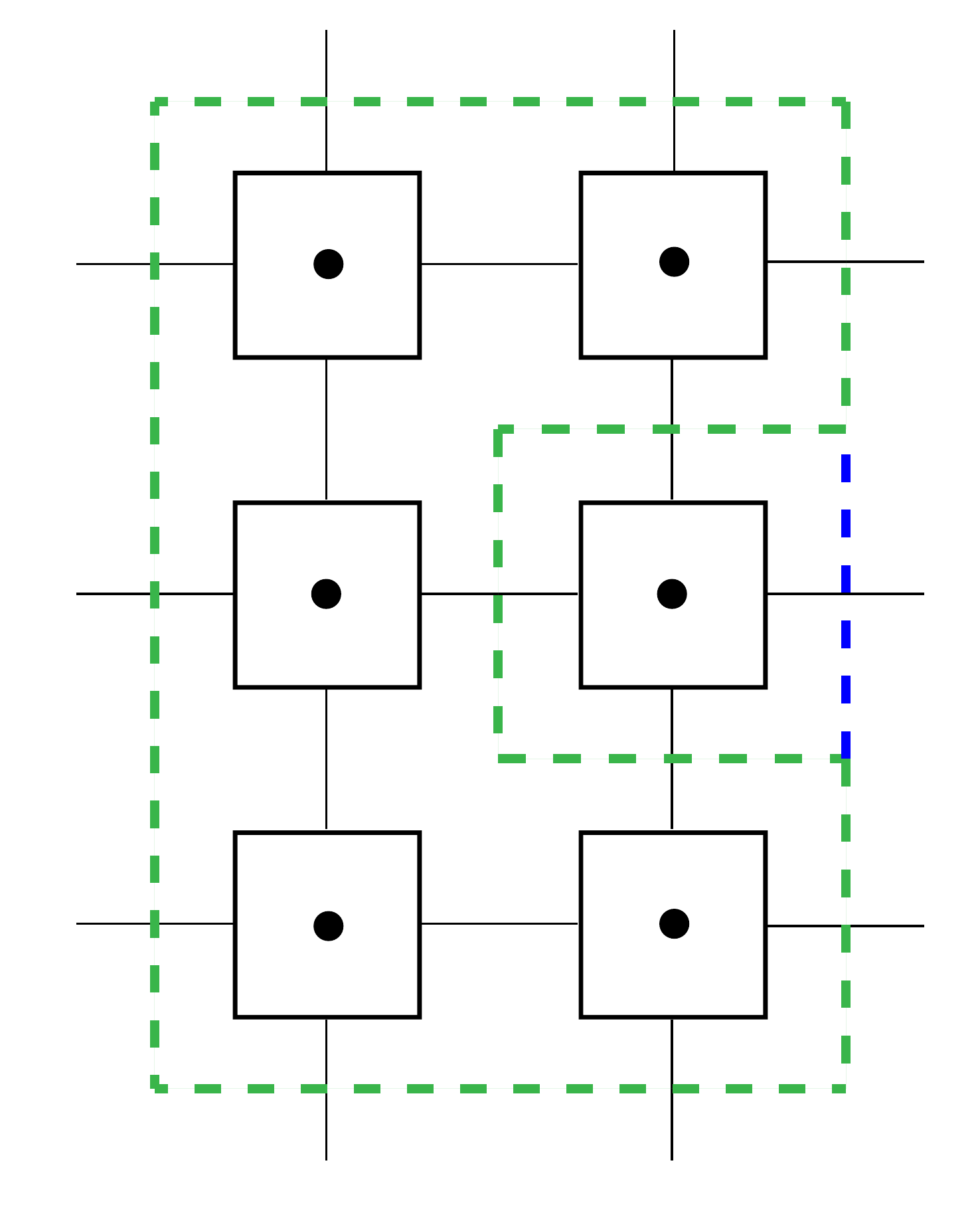}
\end{subfigure}
\caption{\textit{Left:} in this example, the edge of $a$ is a nonconvex surface on a static slice of Minkowski space, and the edge of $E(a)$ is its convex hull. \textit{Right:} a tensor network analogue of the region $a$. For random tensor networks, the map from the boundary of $E(a)$ (blue) to the boundary of $a$ (green) will be an approximate isometry.}
\label{fig:introfig}
\end{figure}

For general regions in time-dependent spacetimes, the correct definition of the generalized entanglement wedge seems less clear. We discuss one possible definition in some detail and show that it satisfies some but not all of the properties satisfied by $E$.

\subsection{Outline}
In Sec.~\ref{sec-motivation} we motivate our prescription. We argue in Sec.~\ref{sec-why} that it must be possible to assign a larger entanglement wedge at least to some gravitating regions. Universality then suggests that all gravitating regions have an entanglement wedge. In Sec.~\ref{sec-tensor} we analyze the example of tensor network toy models and argue that a sensible entanglement wedge may be assigned to bulk regions.

In Sec.~\ref{sec-prescription}, we define the entanglement wedge $E(a)$ of a bulk region $a$. We focus on the ``static'' case, when $a$ and $E(a)$ are part of time-reversal invariant initial data. That is, we propose a generalization of the Ryu-Takayanagi prescription~\cite{Ryu:2006bv} including quantum-corrections~\cite{Faulkner:2013ana}. We establish some definitions and notation in Sec.~\ref{sec-preliminaries}. In Sec.~\ref{sec-gew} we define $E(a)$ as the spatial region with smallest generalized entropy, among all regions that contain $a$. 
In Sec.~\ref{sec-properties} we show that $E$ obeys nontrivial properties consistent with its interpretation as an entanglement wedge: no-cloning, nesting, and strong subadditivity.

In Sec.~\ref{sec-covariant}, we discuss possible generalizations of our prescription to the time-dependent case, analogous to the Hubeny-Rangamani-Takayanagi prescription~\cite{Hubeny:2007xt} and its quantum extension by Engelhardt and Wall~\cite{Engelhardt:2014gca}. After introducing relevant concepts in Sec.~\ref{sec-covdefs}, we consider $E_n(a)$, the wedge with the smallest generalized entropy among all quantum-normal wedges that contain $a$ and share the conformal boundary (if any) of $a$. We show that $E_n$ reduces to the entanglement wedge prescription for boundary regions and static regions. Moreover, $E_n$ satisfies an appropriate no-cloning theorem. However, $E_n$ fails to satisfy strong subadditivity and nesting. We briefly discuss some alternative proposals.

\subsection{Discussion}
\label{sec-discussion}
By analogy with the interpretation of the usual entanglement wedge, our results suggest that information in $E(a)$ can be manipulated or summoned by a bulk observer in $a$. This is a striking implication. 

We expect that such operations would not admit an interpretation in terms of a continuous, classical spacetime. Indeed, even simple operations in AdS, such as the instantaneous, local addition of a particle deep in the bulk, would violate the Bianchi identity and hence cannot be represented as a continuous geometry without introducing time-folds. Yet, this operation can be implemented instantaneously on a global boundary slice $\sigma$. Continuity of the boundary manifold across $\sigma$ then suggests that the bulk exists as a semiclassical spacetime in the past and future of $\sigma$, but not spacelike to $\sigma$. We expect that the reconstruction, from a bulk region $a$, of operators in $E(a)\cap a'$ involves comparably violent breakdowns of the spacetime geometry, at least outside of $a$.

We stress that the full physical significance of our proposal is not yet clear. In particular, we are not proposing a specific reconstruction map for  operators in $E(a)\cap a'$, or any other generalization of the known entries of the holographic dictionary beyond the entanglement wedge prescription. But given its nontrivial properties, we expect $E$ to play a role in formulating quantum gravity as a holographic duality beyond AdS/CFT, for arbitrary spacetimes.

\subsection{Relation to other work} 
\label{sec-otherwork}

In Ref.~\cite{Dong:2020uxp}, a definition was given for entanglement islands of \emph{low-energy} bulk fields in a weakly-gravitating region. The fields have an explicit momentum cut-off, and higher energy operators, such as those that would create black holes, are explicitly excluded. Hence, when applied to an asymptotic bulk region, the prescription~\cite{Dong:2020uxp} would not give the entanglement wedge of the corresponding boundary region. Thus, while we believe that Ref.~\cite{Dong:2020uxp} correctly treats the problem formulated there, its construction cannot be the right answer to the question we pose.

In Ref.~\cite{Grado-White:2020wlb}, a ``restricted maximin prescription'' was proposed for the entanglement wedge of a boundary region in a \emph{cut-off} spacetime. Again, this prescription is intended to answer a different question to the one we are interested in here. As we describe in  Sec.~\ref{sec-covariantdiscussion}, a closely related prescription --- with the edge of $a$ playing the role of the cut-off surface --- yields a time-dependent generalization of our proposal that obeys strong subadditivity, but which violates the nesting property required for the entanglement wedges of bulk regions.

Finally, the present work does not support claims that there is no Page curve for an evaporating black hole \cite{Raju:2021lwh}. The Page curve is expected to describe the entropy of low-energy Hawking radiation modes, not all the information in principle accessible far from the black hole. For example, let $a$ be the exterior of a large sphere in asymptotically flat space. Then the generalized entanglement wedge we define, $E(a)$, can include the entire interior enclosed by $a$. But this does not mean that the information in $E(a)\backslash a$ can be accessed via low-energy operators in $a$.

\section{Motivation}
\label{sec-motivation}

\subsection{Why a Bulk Region Should Have an Entanglement Wedge}
\label{sec-why}

We begin by arguing that at least {\em some} gravitating bulk regions constitute legitimate input to the entanglement wedge prescription. We will discuss the two examples shown in Fig.~\ref{fig-motivate}.

First, consider the Hawking radiation emitted by a black hole in asymptotically flat space, after the Page time. The radiation must contain the same information, whether it has been extracted into an auxiliary nongravitating system, or arrived at a portion $A$ of future null infinity, or is still traversing a distant, weakly gravitating region $a$. Operationally, the presence of an island $b$ in the former cases indicates unitarity: sufficiently careful experiments would show that the ultimate state of the radiation is pure. But if such experiments could only succeed in the complete absence of gravity, they would fail in the real world, so the question of unitarity would be operationally meaningless. 

These considerations make it implausible that the radiation has an entanglement island only if its self-gravity is completely turned off, but not if its self-gravity is arbitrarily small. It follows that the weakly gravitating region $a$ must be assigned an entanglement wedge that includes the island $b$:
\begin{equation}
    E(\mbox{distant Hawking radiation}~a) \supset a \cup b~.
\end{equation}

Let us pause for a moment to examine what is happening here. Careful measurements on one part of the bulk spacetime, the Hawking radiation, are allowing us to extract information that resides (from the semiclassical viewpoint) in a distant spacelike-separated region: the black hole interior. This is a striking example of a fundamental nonlocality that appears to be present in quantum gravity; sufficiently complex operators do not have to respect the semiclassical structure of spacetime. Importantly this fundamental nonlocality is not solely a feature of black hole physics; instead it is an essential aspect of holography. This brings us to our second example (Fig.~\ref{fig-motivate}, right panel). 

In the AdS/CFT correspondence, local CFT operators in a conformal boundary subregion $A$ are dual to (quasi-)local bulk operators near $A$~\cite{Banks:1998dd,Hamilton:2006az}. The notion of ``near'' can be made precise by defining a bulk region $a$ as the union of the entanglement wedges of tiny boundary regions containing slightly smeared local boundary operators.\footnote{High-dimensional local operators in the CFT will not correspond to low-energy fields in the bulk effective field theory. Instead, they will generically create large black holes. However these black holes will still be localized near the asymptotic boundary. See, \emph{e.g.}, the discussion in \cite{Harlow:2018tng}.} But local CFT operators generate the entire algebra of the CFT. It follows that the algebra generated by {\em bulk operators} in the near-boundary bulk region $a$ encodes the \emph{entire} entanglement wedge of $A$.\footnote{In 2+1 boundary dimensions and higher, Wilson loops play an important role; these are nonlocal operators that (in non-Abelian gauge theories) cannot be written as a product of gauge-invariant local operators. However, any Wilson loop will be contained in a thin boundary strip with arbitrarily shallow entanglement wedge. Therefore it should also be accessible as a bulk operator in the corresponding near-boundary region, and should still be contained in the associated bulk algebra.} Information apparently located deep in the bulk must be secretly, and highly nonlocally, encoded in quasilocal bulk operators near the asymptotic boundary.\footnote{In fact, this encoding is exponentially simpler than the encoding of the black hole interior in Hawking radiation \cite{Harlow:2013tf, Brown:2019rox}.} 

Thus, the traditional entanglement wedge prescription for boundary subregions, $\mathrm{EW}(A)$, can be reinterpreted as a statement about the algebra generated by quasi-local bulk operators in the asymptotic bulk region $a$, without any reference to a CFT dual. It follows that an asymptotic bulk region $a$ in asymptotically AdS spacetime must possess an entanglement wedge --- indeed, the same entanglement wedge that would have been assigned to its conformal boundary $A$:
\begin{equation} \label{eq:asymptoticspecialcase}
    E(\mbox{asymptotic AdS bulk region}~a) = \mathrm{EW}(A)~.
\end{equation}

If two such apparently dissimilar bulk regions --- Hawking radiation and asymptotic regions in AdS --- encode a larger entanglement wedge, then Occam's razor suggests that
\begin{itemize}
    \item \emph{any} bulk region should have an associated entanglement wedge; and 
    \item \emph{only} bulk regions possess an entanglement wedge. (Boundary regions and auxiliary systems can be viewed as idealized limits of bulk regions that are near asymptotic infinity or very weakly gravitating respectively.)
\end{itemize}

\subsection{Motivation From Tensor Networks }
\label{sec-tensor}

\begin{figure}[t]
\begin{subfigure}{.48\textwidth}
  \centering
 \includegraphics[width = 0.75\linewidth]{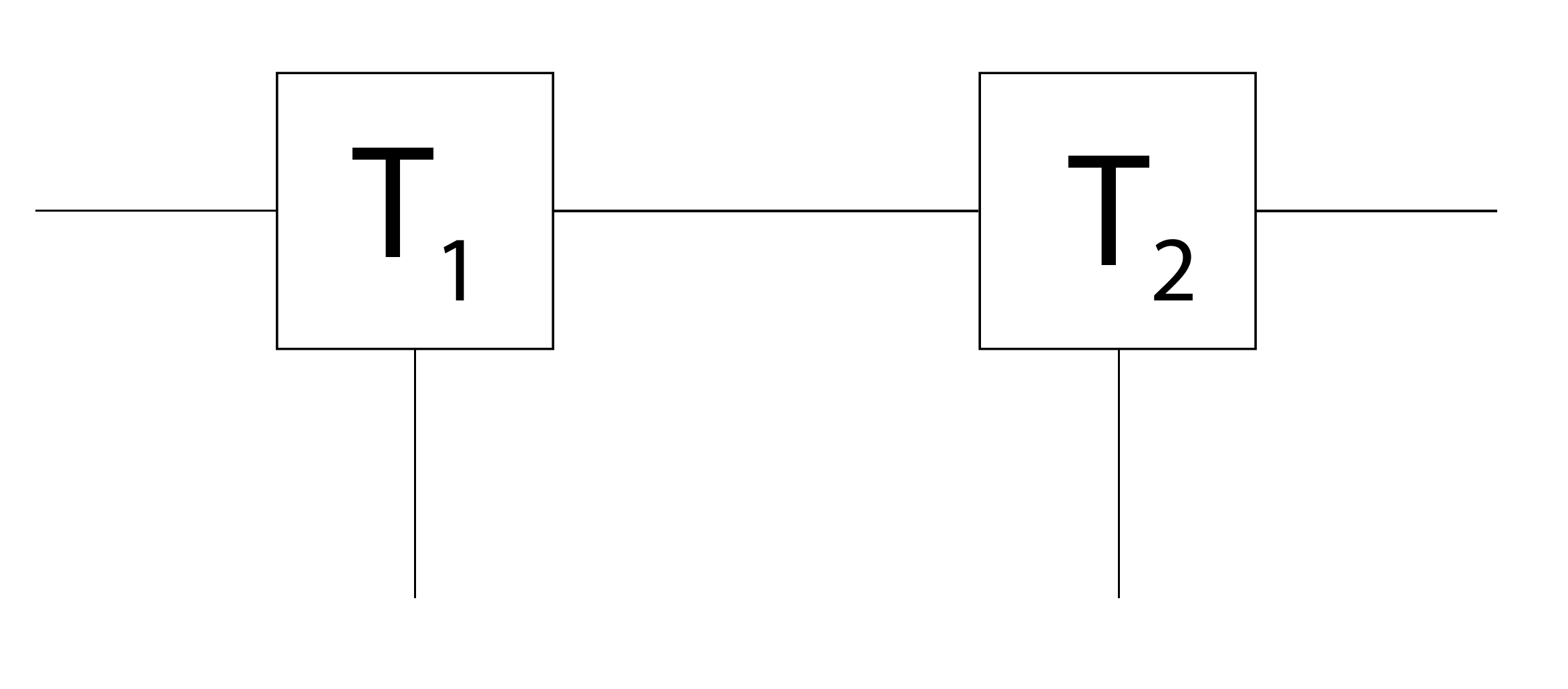}
\end{subfigure}
\begin{subfigure}{.48\textwidth}
  \centering
 \includegraphics[width = 0.75\linewidth]{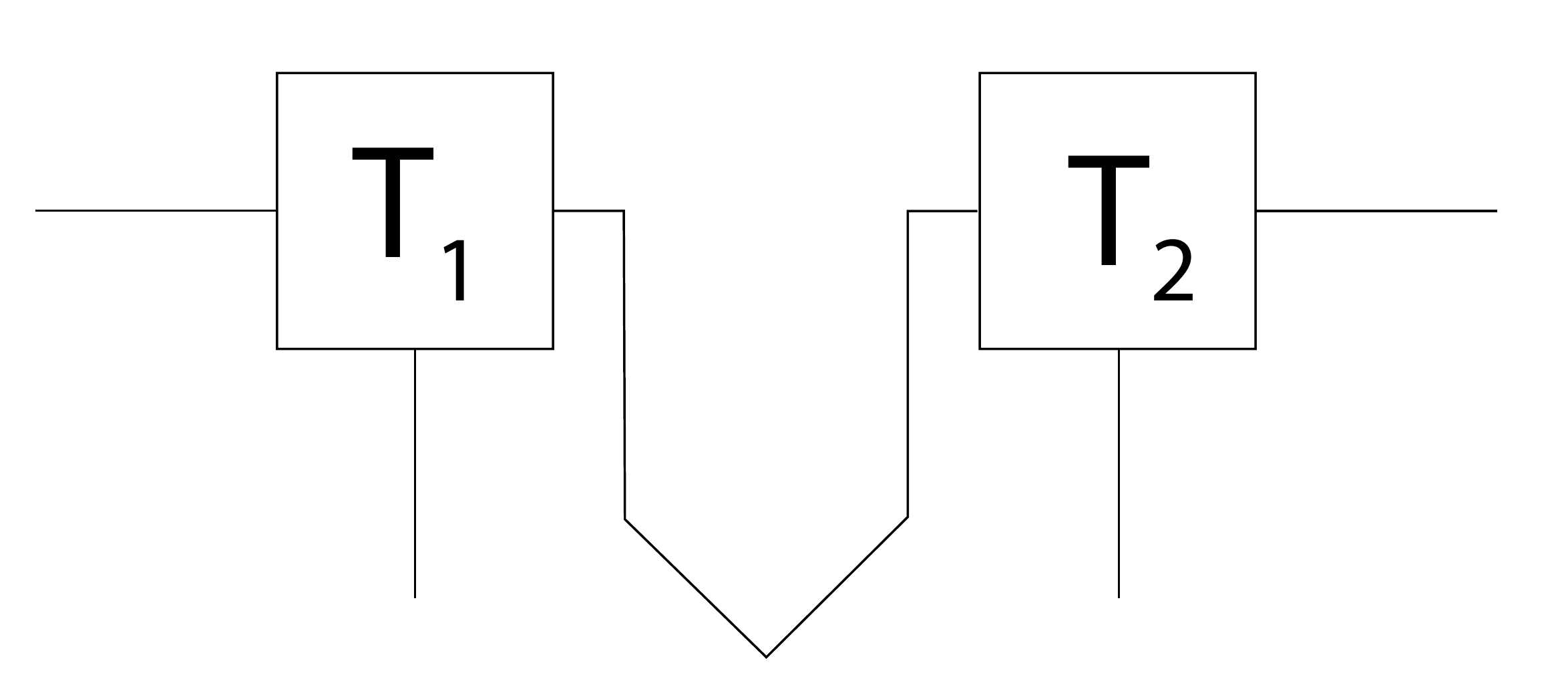}
\end{subfigure}
\caption{\textit{Left:} Two tensors $T_1$ and $T_2$ in a tensor network. The in-plane legs (horizontal) connect different tensors together, while the bulk legs (vertical) input matter degrees of freedom. \textit{Right:} By bending the in-plane leg connecting $T_1$ and $T_2$ around, we can reinterpret it  as two maximally entangled bulk legs, one on each tensor.}
\label{fig-tensors}
\end{figure}

Additional motivation for such a proposal comes from tensor networks, a useful toy model of static states in quantum gravity~\cite{Swingle:2009bg, Pastawski:2015qua, Hayden:2016cfa, Bao:2018pvs}. In these models, each tensor represents a patch of the gravitational spacetime, with the logarithm of the dimension of an edge capturing the area (in Planck units) of the surface connecting two neighbouring patches. Bulk degrees of freedom within the patch are described by additional ``out-of-plane'' legs that feed into the network.

Importantly, the distinction between out-of-plane and in-plane tensor network legs is not precise. Mathematically, highly entangled Rindler modes on either side of a cut play an identical role in the tensor network as in-plane legs representing area, as shown in Fig.~\ref{fig-tensors}. The same effect is seen in gravity where the Rindler modes and geometrical area are equivalent under renormalization group flow. As we increase the UV cut-off on the bulk quantum field theory, the amount of Rindler entanglement across a surface increases; this increase is cancelled in a computation of generalized entropy by the decrease in $A/4G$ created by the renormalization of Newton's constant.

If we extrapolate the UV cut-off all the way to the Planck scale, it is natural to conjecture that the area term, or analogously the in-plane tensor network legs, go away entirely, with the entanglement coming entirely from out-of-plane bulk legs. In such a limit, the tensor network consists purely of a set of projection operators that map the bulk state to the boundary state.

Why does this suggest that a region $a$ encodes the larger region $E(a)$, as suggested in Fig.~\ref{fig:introfig}? If in-plane legs are really bulk Planck scale bulk legs, then an operator $O$ can be reconstructed by a (potentially very high energy) operator within the region $a$ if the action of $O$ is the same as an operator $\tilde O$ that acts on the out-of-plane legs within $a$, along with the in-plane legs at the boundary of $a$. After all, both sets of degrees of freedom should be in principle accessible from $a$, even if the Planckian degrees of freedom %are difficult to manipulate in practice.
may not be described by an effective field theory.

In a tensor network, $E(a)$ corresponds to the smallest-entropy cut through the network such that $a$ is contained entirely within the cut (see Fig. \ref{fig:introfig}, right panel). Let us assume for simplicity that the dominant contribution to the entropy of any given cut comes from the maximally entangled in-plane legs. (In gravity, this corresponds to taking the semiclassical limit where the area term gives the dominant contribution to generalized entropy.) Now consider the map $V$ induced by the tensor network from the boundary of $E(a)$, together with bulk legs in $E(a) \backslash a$, to the boundary of $a$. By definition, any intermediate cut through this subnetwork describing $V$ has larger area (i.e. much larger bond dimension) than the input to $V$. For networks with sufficiently generic tensors, this ensures that the map $V$ is an (approximate) isometry,\footnote{Concretely, this will be the case with very high probability if the tensors are drawn from a Haar random ensemble \cite{Hayden:2016cfa}.} and hence that the operator $\tilde O = V^\dagger O V$, which acts on the in-plane legs at the boundary of $a$, successfully reconstructs any operator $O$ acting in $E(a) \backslash a$.

\section{Generalized Entanglement Wedges in Time-reflection\\ Symmetric Slices}
\label{sec-prescription}

\subsection{Preliminaries}
\label{sec-preliminaries}

\begin{defn}\label{def:trs}
A spacetime $M$ is \emph{time-reflection symmetric} if there exists a $\mathbb{Z}_2$ action on $M$ that exchanges past and future timelike directions and preserves all points in a Cauchy surface $\Sigma$. %We use $a,b,\ldots$ to denote \emph{time-reflection symmetric wedges}, defined as wedges satisfying $T(a) = a$.
\end{defn}

\begin{conv}
In this section, we will only be concerned with the time-reflection symmetric Cauchy surface $\Sigma$, regarded as a manifold with Riemannian metric $h$. We assume that $\Sigma$ is inextendible.
\end{conv}

\begin{defn}
We use $\partial s$ to denote the boundary of a set $s\subset \Sigma$ in the induced topology of $\Sigma$. Also $\cl s\equiv s \cup \partial s$ and $\setint s = s \setminus \partial s$. \end{defn}

\begin{defn}
A \emph{wedge} $a$ is any open subset of $\Sigma$ that is the interior of its closure: $a=\setint \cl a$. (The name is slightly more natural in the more general time-dependent context; it is adopted for compatibility with the existing phrase ``entanglement wedge.'')
\end{defn}

\begin{defn}
The intersection of two wedges is a wedge; but the union, complement, and relative complement need not be. Given two wedges $a$ and $b$, the \emph{wedge union}, \emph{wedge complement}, and \emph{wedge relative complement} are wedges defined by 
\begin{align}
    a\Cup b & \equiv \setint\cl(a\cup b)~,\\
    a' & \equiv \setint a^C~, \\
    a\setminus b & \equiv a \cap b'~.
\end{align}
\end{defn}

\begin{defn}\label{def:area}
The \emph{area} of a wedge $a$ is the geometric area of $\partial a$ and is denoted $\A(a)$.
\end{defn}

\begin{defn}\label{def:sgen}
The {\em generalized entropy} of a wedge $a$ is
\begin{equation}\label{sgendef}
    S_{\rm gen}(a)\equiv \frac{\A(\partial a)}{4G\hbar} + S(a) + \ldots
\end{equation}
Here %${\A}$ denotes the area, and
\begin{equation}
    S(a) = - \tr \rho_a \log \rho_a
\end{equation}
 is the von Neumann entropy of
\begin{equation}
     \rho_a = \tr_{a'} \rho~,
\end{equation}
the density operator of the quantum fields restricted to $a$.
\end{defn}
\begin{rem}
The area in Eq.~\eqref{sgendef} can be thought of as a counterterm that cancels the leading divergence in the von Neumann entropy $S$. Additional counterterms for subleading divergences have been omitted. See the appendix of Ref.~\cite{Bousso:2015mna} for details.
\end{rem}

\begin{defn}
Let $\tilde \Sigma\equiv (M,\tilde h)$ denote the \emph{conformal completion} of $\Sigma$, obtained by conformally mapping $\Sigma$ to a subset of a compact set (for example, a sphere) and taking the closure of the image. $\tilde \Sigma$ may be a manifold-with-boundary~\cite{Wald:1984rg}. The boundary of the image, $\partial \tilde \Sigma$, is called the \emph{conformal boundary} of $\Sigma$. We will not distinguish notationally between a set $a\subset \Sigma$ and its image in $\tilde\Sigma$. However, such a distinction is crucial for boundaries of sets. 
\end{defn}

\begin{defn}\label{def:staticconformaledge}
The boundary of a wedge $a\subset \tilde \Sigma$ will be denoted $\delta a$. As shown in Fig. \ref{fig-conf-bdy}, we define the {\em conformal boundary} of $a$ as the set
\begin{equation}
    \tilde\partial a\equiv \delta a\cap \partial\tilde \Sigma~;
\end{equation}
thus,
\begin{equation}
    \delta a = \partial a\sqcup \tilde\partial a~.
\end{equation}
\end{defn}

\begin{figure}[t]
  \centering
 \includegraphics[width = 0.5\linewidth]{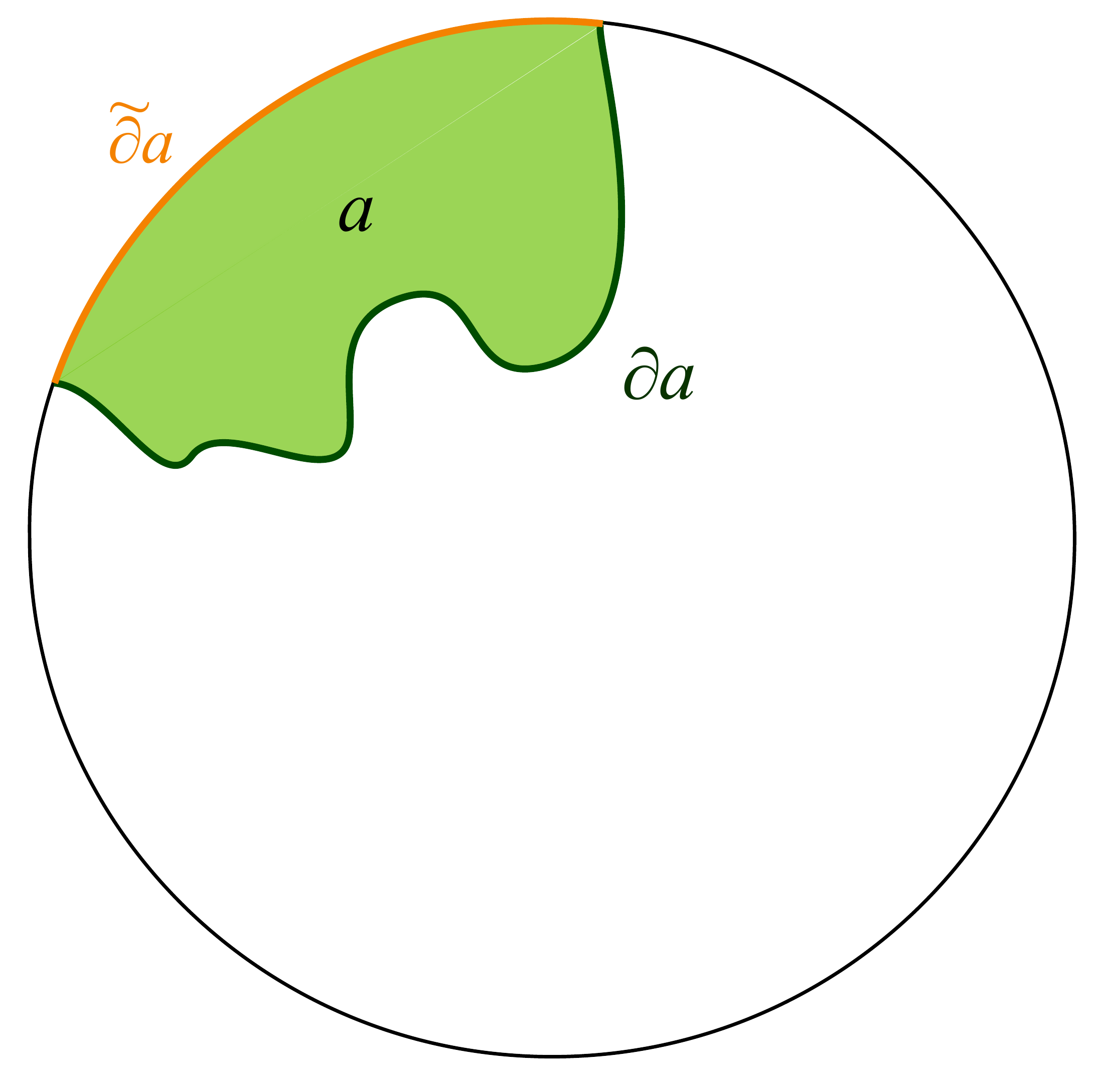}
\caption{The boundary $\delta a \in \tilde\Sigma$ of a wedge $a$ in the conformal completion $\tilde\Sigma$ decomposes into the original boundary $\partial a \in \Sigma$ and a conformal boundary $\tilde \partial a \in \partial \tilde\Sigma$.}
\label{fig-conf-bdy}
\end{figure}

\subsection{Definition and Basic Properties}
\label{sec-gew}

\begin{defn}\label{ewstatic}
Given a wedge $a\subset\Sigma$, we define its {\em generalized entanglement wedge}, $E(a)$, as the wedge that satisfies
\begin{equation}
    a\subset E(a)\subset \Sigma~~~\mbox{and}~~~\tilde\partial a = \tilde\partial E(a)
\end{equation}
and which has the smallest generalized entropy among all such wedges. We assume without proof that $E(a)$ exists, and for convenience we shall assume that it is unique.
\end{defn}

\begin{lem}[Monotonicity]\label{monolem}
Let $a$ and $b$ be wedges with the same conformal boundary, $\tilde\partial a=\tilde\partial b$. Then 
\begin{equation}
    a\subset b\implies \S[E(a)]\leq \S[E(b)]~.
\end{equation}
That is, the generalized entropy of the entanglement wedge increases monotonically with the input wedge under inclusion, if the conformal boundary is held fixed.
\end{lem}
\begin{proof}
$E(b)\supset b\supset a \implies \S[E(b)]\geq \S[E(a)]$ by Def.~\ref{ewstatic}.
\end{proof}

\begin{defn}[Entanglement wedge of a boundary region] \label{defn:bdyews}
Let $B \subset \partial \tilde\Sigma$ be a subregion of the conformal boundary  $\partial \tilde\Sigma$. The (ordinary) \emph{entanglement wedge} $\mathrm{EW}(B)\subset \Sigma$~\cite{Ryu:2006bv,Faulkner:2013ana} is the wedge with conformal boundary $B$ and smallest generalized entropy among all such sets.
\end{defn}

\begin{lem}[$\mathrm{EW}$ as a special case of $E$] \label{lem:genew=ews}
If the wedge $a$ lies in the (ordinary) entanglement wedge of its conformal boundary, $a \subset \mathrm{EW}(\tilde\partial a)$, then its generalized entanglement wedge is
\begin{equation}
E(a) = \mathrm{EW}(\tilde\partial a).
\end{equation}
\end{lem}
\begin{proof}
The result follows immediately from Definitions~\ref{ewstatic} and \ref{defn:bdyews}.
\end{proof}

\begin{rem} [Asymptotic bulk regions]
Let $a_n$ be an infinite sequence of wedges with $a_{n+1} \subset a_n$, $\tilde\partial a_{n+1} = \tilde \partial a_n$ and $\cap_n a_n = \varnothing$. Then $E(a_n) = \mathrm{EW}(\tilde \partial a_n)$ for all sufficiently large $n$. 
\end{rem}

\subsection{No-Cloning, Strong Subadditivity, and Nesting}
\label{sec-properties}

\begin{figure}[t]
\begin{center}
  \includegraphics[width = 0.6\linewidth]{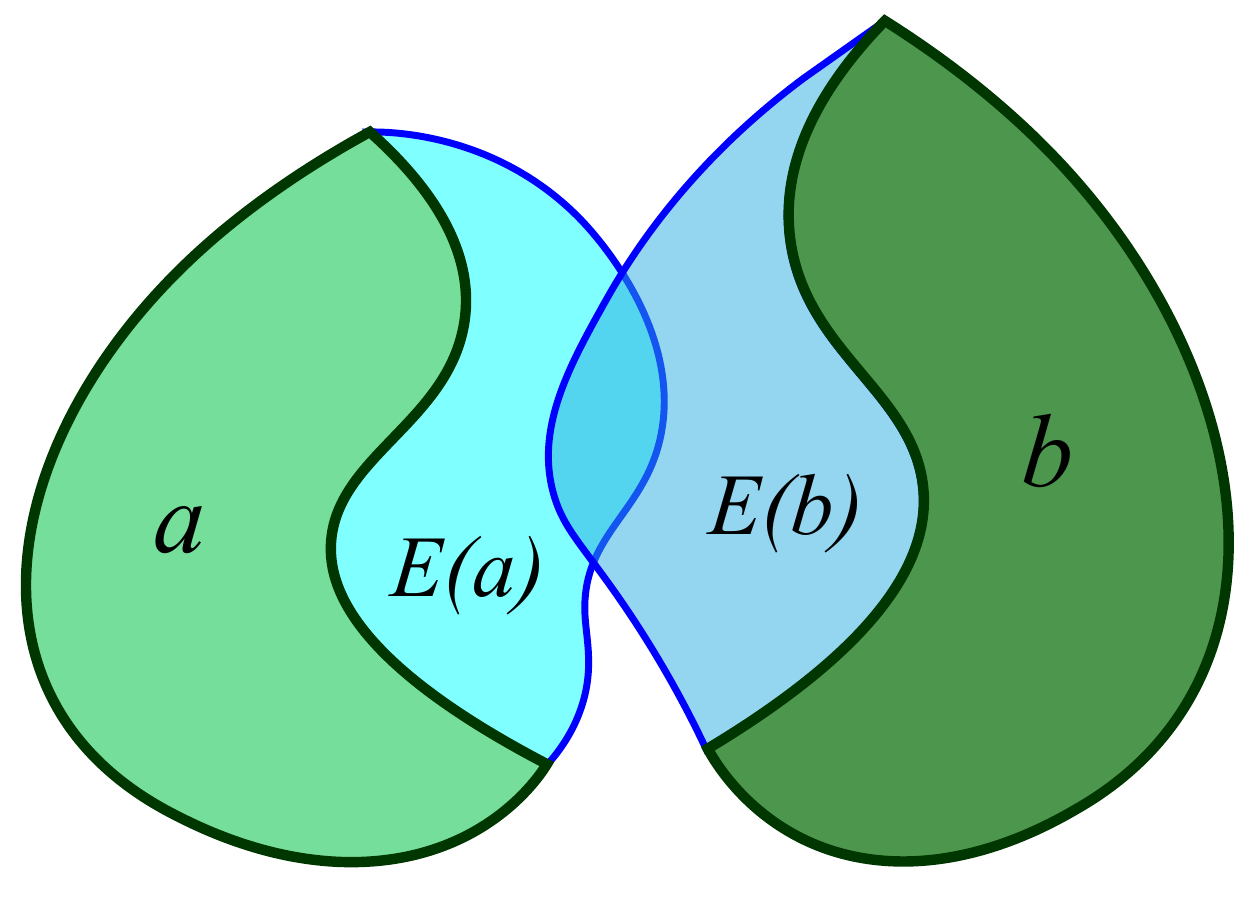} 
\end{center}
\caption{A configuration excluded by Theorem \ref{staticnocloning} (No-cloning): if $a \cap E(b)=\varnothing$ and $b\cap E(a)=\varnothing$ then $E(a) \cap E(b)=\varnothing$, because otherwise removing the overlap from each entanglement wedge would fail to increase their combined generalized entropy.}
\label{fig:nocloning}
\end{figure}

\begin{thm}[No cloning]\label{staticnocloning}
Let $a,b$ be wedges that satisfy
\begin{equation}\label{staticexclusions}
    a \cap E(b)=\varnothing~, ~~ b\cap E(a)=\varnothing~.
\end{equation}
Then
\begin{equation}\label{staticnocloneq}
    E(a)\cap E(b)=\varnothing~.
\end{equation}
\end{thm}

\begin{proof}
As shown in Fig. \ref{fig:nocloning},
\begin{align}
    \A[E(a)\setminus E(b)] + \A[E(b)\setminus E(a)] \leq \A[E(a)] + \A[E(b)]~.
\end{align}
(This need not be an equality; for example, there may be a disconnected shared boundary in $E(a)\cap E(b)$.) Strong subadditivity implies the same inequality for the von Neumann entropies; hence
\begin{align}
    \S[E(a)\setminus E(b)] + \S[E(b)\setminus E(a)] \leq \S[E(a)] + \S[E(b)]~.
\end{align}
That is, the generalized entropy of $E(a)$ or $E(b)$ decreases when $E(a)\cap E(b)$ is removed. By the assumption of the Lemma, $a \subset E(a) \setminus E(b)$ and $b \subset E(b) \setminus E(a)$, so this contradicts Def.~\ref{ewstatic} unless $E(a)\cap E(b)=\varnothing$.
\end{proof}

\begin{figure}[t]
\begin{center}
  \includegraphics[width = 0.6\linewidth]{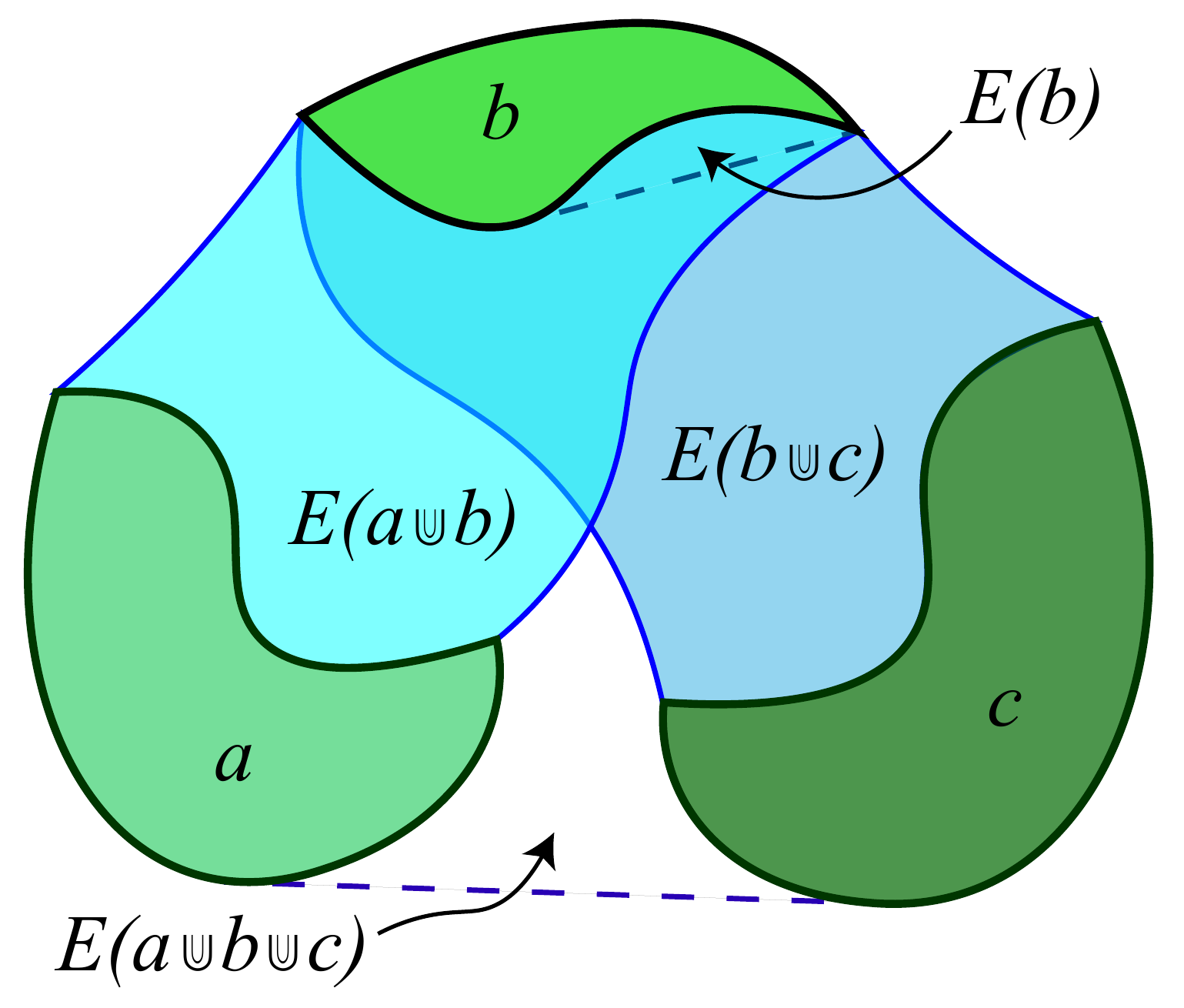} 
\end{center}
\caption{An illustration of the setup from Theorem \ref{thm:ssastatic} (Strong subadditivity): The combined $\S$ of $E(a \Cup b)$ and $E(b \Cup c)$ is at least big as the combined $\S$ of their intersection and union. In turn, these upper bound the combined $\S$ of $E(b)$ and $E(a \Cup b \Cup c)$ respectively.}
\label{fig:SSA}
\end{figure}
\begin{thm}[Strong subadditivity]\label{thm:ssastatic}
Let $a$, $b$ and $c$ be mutually disjoint open subsets of $\Sigma$. Then
\begin{equation}\label{ssastatic}
    \S[E(a \Cup b)] + \S[E(b \Cup c)] \geq \S[E(b)] + \S[E(a \Cup b \Cup c)]~.
\end{equation}
\end{thm}
\begin{proof}
The setup is illustrated in Fig. \ref{fig:SSA}. Rearranging components of surfaces, we have
\begin{align}
    \A[E(a \Cup b)] & + \A[E(b \Cup c)] \geq \nonumber \\
    & \A[E(a \Cup b)\cap E(b \Cup c)] + \A[E(a \Cup b)\Cup E(b \Cup c)]~.
\end{align}
(This need not be an equality since the wedge union can erase boundary portions.) Strong subadditivity implies the same inequality for the von Neumann entropies, and hence for the generalized entropy: 
\begin{equation}
   \S[E(a \Cup b)] + \S[E(b \Cup c)]\geq  
   \S[E(a \Cup b)\cap E(b \Cup c)] + \S[E(a \Cup b)\Cup E(b \Cup c)]
\end{equation}
The first set on the right contains $b$ and the second contains $a\Cup b\Cup c$, so Def.~\ref{ewstatic} implies Eq.~\eqref{ssastatic}.
\end{proof}

\begin{thm}[Nesting]
\begin{equation}\label{thm:ewnstatic}
    a\subset b \implies E(a)\subset E(b)~.
\end{equation}
\end{thm}
\begin{figure}[t]
\begin{center}
  \includegraphics[width = 0.4\linewidth]{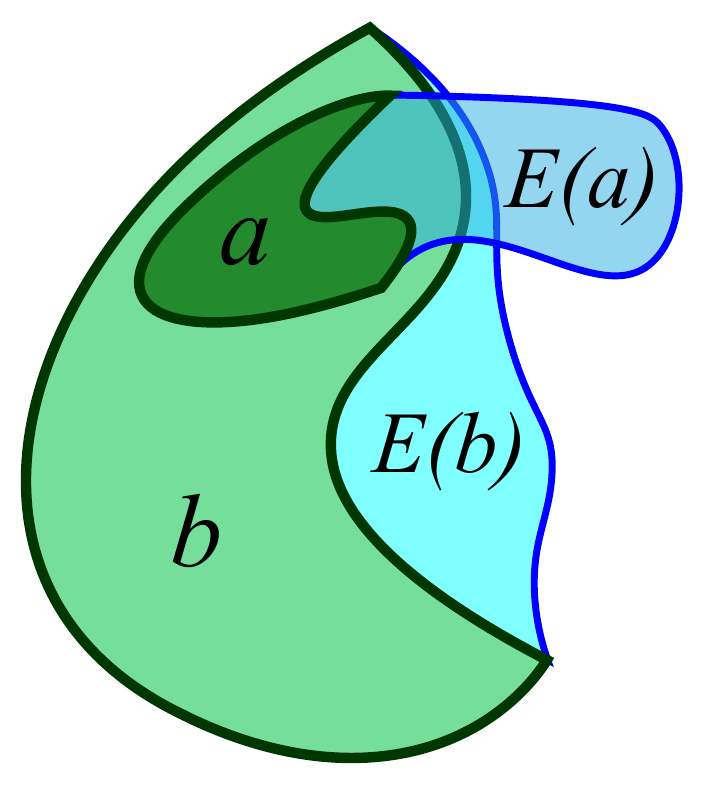} 
\end{center}
\caption{A hypothetical counterexample to Theorem \ref{thm:ewnstatic} (Nesting): if $a \subset b$ but $E(a) \not\subset E(b)$, we can replace the two entanglement wedges by $E(a) \cap E(b)$ and $E(a) \Cup E(b)$ respectively, decreasing their combined area.}
\label{fig:nesting}
\end{figure}
\begin{proof}
The setup is illustrated in Fig. \ref{fig:nesting}. Since $a\subset b$, we have $a\subset E(a)\cap E(b)$, so by Def.~\ref{ewstatic}
\begin{align}\label{ewnstaticproof1}
    \S[E(a)] & \leq \S[E(a)\cap E(b)]~,\\
    \S[E(b)] & \leq \S[E(a)\Cup E(b)]~.
\end{align}
But
\begin{equation}
    \A[E(b)] + \A[E(a)] \geq \A[E(a)\cap E(b)] + \A[E(a)\Cup E(b)]~. 
\end{equation}
(This need not be an inequality since the wedge union can erase boundary portions.) Strong subadditivity implies the same inequality for the von Neumann entropies, and hence for the generalized entropy:
\begin{equation}
    \S[E(b)] + \S[E(a)] \geq \S[E(a)\cap E(b)] + \S[E(a)\Cup E(b)]~.
\end{equation}
To avoid a contradiction, all three of the above inequalities must be saturated. The assumed uniqueness of $E(a)$ implies $E(a)\cap E(b)=E(a)$, and hence $E(a)\subset E(b)$.
\end{proof}

\section{Time-dependence}
\label{sec-covariant}

In this Section we discuss possible generalizations of our proposal to the case where $a$ does not lie on a time-reflection symmetric Cauchy surface.
%In this appendix, we explore a possible covariant extension of the generalized entanglement wedge prescription. We do not present a final answer. 
We will sketch approaches, obtain some partial results, and outline specific challenges. We focus primarily on one possible generalization -- the smallest-generalized-entropy quantum normal wedge $E_n(a)$ -- which turns out to obey a no-cloning theorem but not strong subadditivity or nesting. We briefly comment on other possible generalizations that preserve strong subadditivity and nesting respectively, but do not find any single definition that preserves all three.

\subsection{Preliminaries}
\label{sec-covdefs}

We begin by introducing natural generalizations of the key objects introduced in Sec.~\ref{sec-preliminaries}. Let $M$ be a globally hyperbolic Lorentzian spacetime with metric $g$. The chronological and causal future and past, $I^\pm$ and $J^\pm$; and the future and past domains of dependence and Cauchy horizons, $D^\pm$ and $H^\pm$; are defined as in Wald~\cite{Wald:1984rg}. Given $s\subset M$, we use $\partial s$ to denote the boundary of $s$ in $M$, and $s'$ to denote the interior of the set of points that are spacelike related to all points in $s$, that is, points outside the causal future and past of $s$.

\begin{defn}\label{def:covwedge}
A {\em wedge} is a set $a\subset M$ that satisfies $a=a''$ (see Fig. \ref{fig-wedges}, left). 
\end{defn}
\begin{figure}[t]
\begin{subfigure}{.48\textwidth}
  \centering
 \includegraphics[width = 0.8\linewidth]{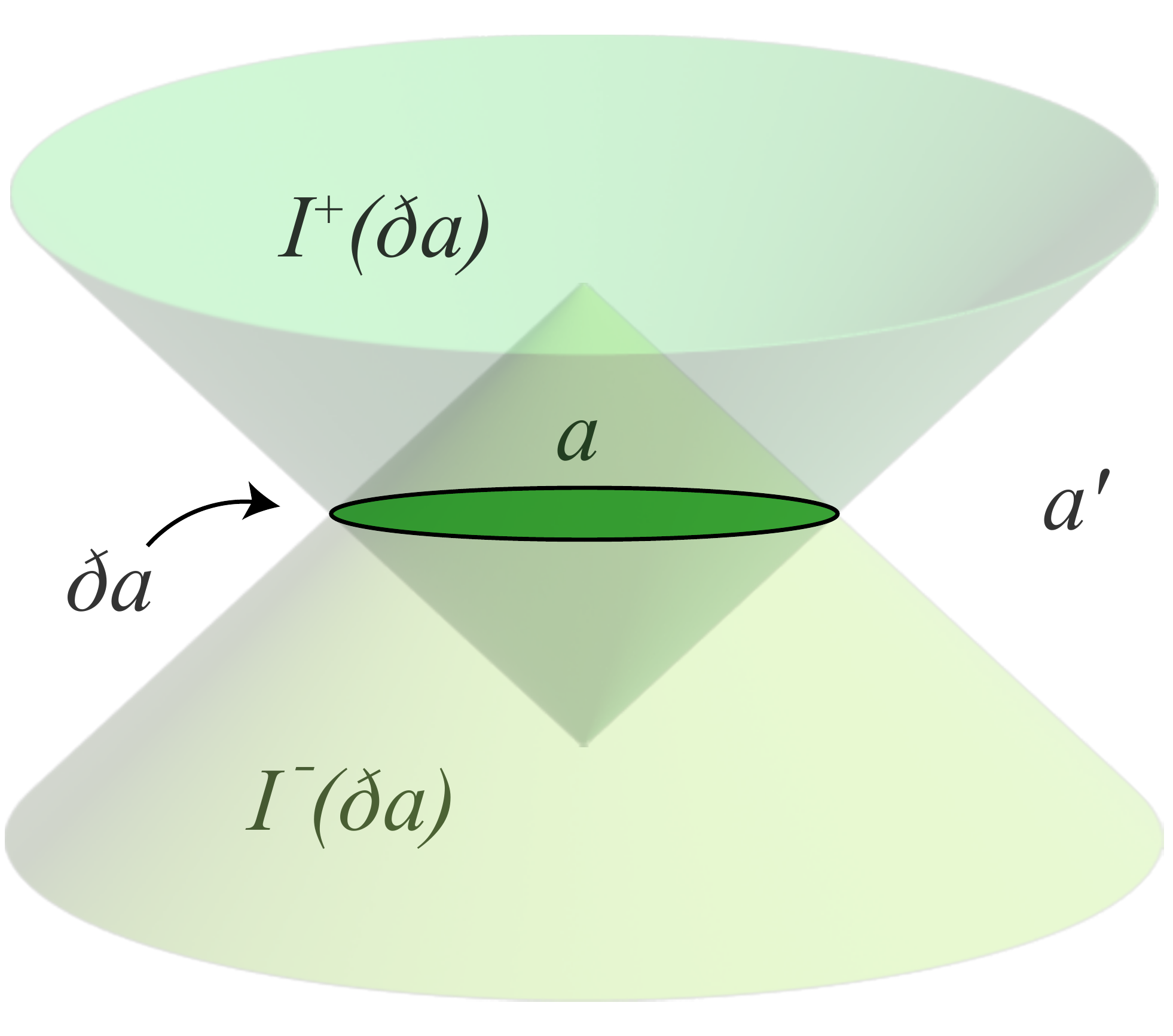}
\end{subfigure}
\begin{subfigure}{.48\textwidth}
  \centering
 \includegraphics[width = 0.8\linewidth]{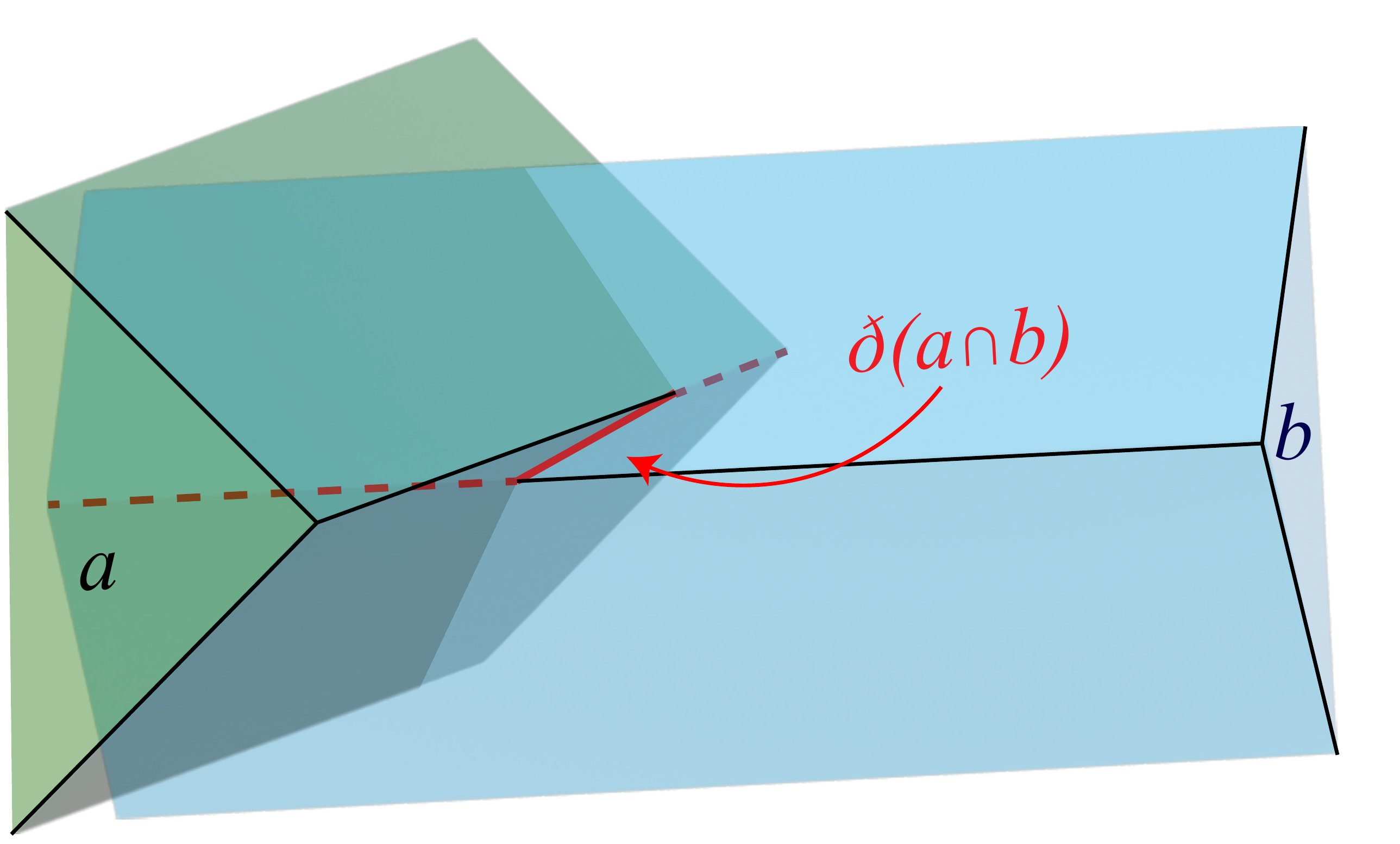}
\end{subfigure}
\caption{\emph{Left:} A spherical wedge $a$, its complement wedge $a'$, and their shared edge $\eth a$ in Minkowski space. The wedge $a$ is diamond-shaped in this spacetime diagram; a Cauchy slice of $a$ is shown in dark green. This wedge is ``normal,'' since outgoing orthogonal lightrays expand.  \emph{Right:} The intersection of two wedges is itself a wedge, with an edge that decomposes as $\eth(a \cap b) = \cl\left[(\eth a\cap b) \sqcup (H^+(a)\cap H^-(b))\sqcup \set{a\leftrightarrow b}\right]$.}
\label{fig-wedges}
\end{figure}

\begin{rem}\label{wilem}
It can be shown that the intersection of two wedges $a,b$ is a wedge (see Fig. \ref{fig-wedges}, right); and the \emph{complement wedge} $a'$ is a wedge:
\begin{equation}
    (a\cap b)'' = a\cap b~;~~a'''=a'~.
\end{equation}
\end{rem}

\begin{defn}
Given two wedges $a$ and $b$, we define the {\em relative complement of $b$ in $a$} as $a\cap b'$, and the {\em wedge union} as $a\Cup b\equiv (a'\cap b')'$ (see Fig. \ref{fig:union}); these objects are wedges by the above remark.
\end{defn}
\begin{rem}%%%\label{edgeunion}
The wedge union satisfies $a\Cup b \supset a\cup b$. It is minimal in the sense that any wedge that contains $a\cup b$ must contain $a\Cup b$.
\end{rem}
\begin{defn}%%%\label{def:edge}
The {\em edge} $\eth a$ of a wedge $a$ is 
defined by $\eth a\equiv \partial a \cap \partial a'$. Note that $a$ is fully characterized by specifying $\eth a$ and one spatial side of $\eth a$.
\end{defn}

\begin{defn}
The \emph{area} and \emph{generalized entropy} of a wedge $a$ are defined as in the static case, Defs.~\ref{def:area} and \ref{def:sgen}, with $\partial a$ replaced by $\eth a$. Note that neither depend on a choice of Cauchy slice of $a$.
\end{defn}

\begin{defn}\label{def:conformaledge}
Given a wedge $a$, we distinguish between its edge $\eth a$ in $M$ and its edge $\deltabar a$ as a subset of the conformal completion $\tilde M$. The latter can contain an additional piece, the {\em conformal edge}
\begin{equation}
    \tilde\eth a\equiv \deltabar a\cap \partial\tilde M~.
\end{equation}
\end{defn}

\begin{figure}[t]
\begin{center}
  \includegraphics[width=0.8\linewidth]{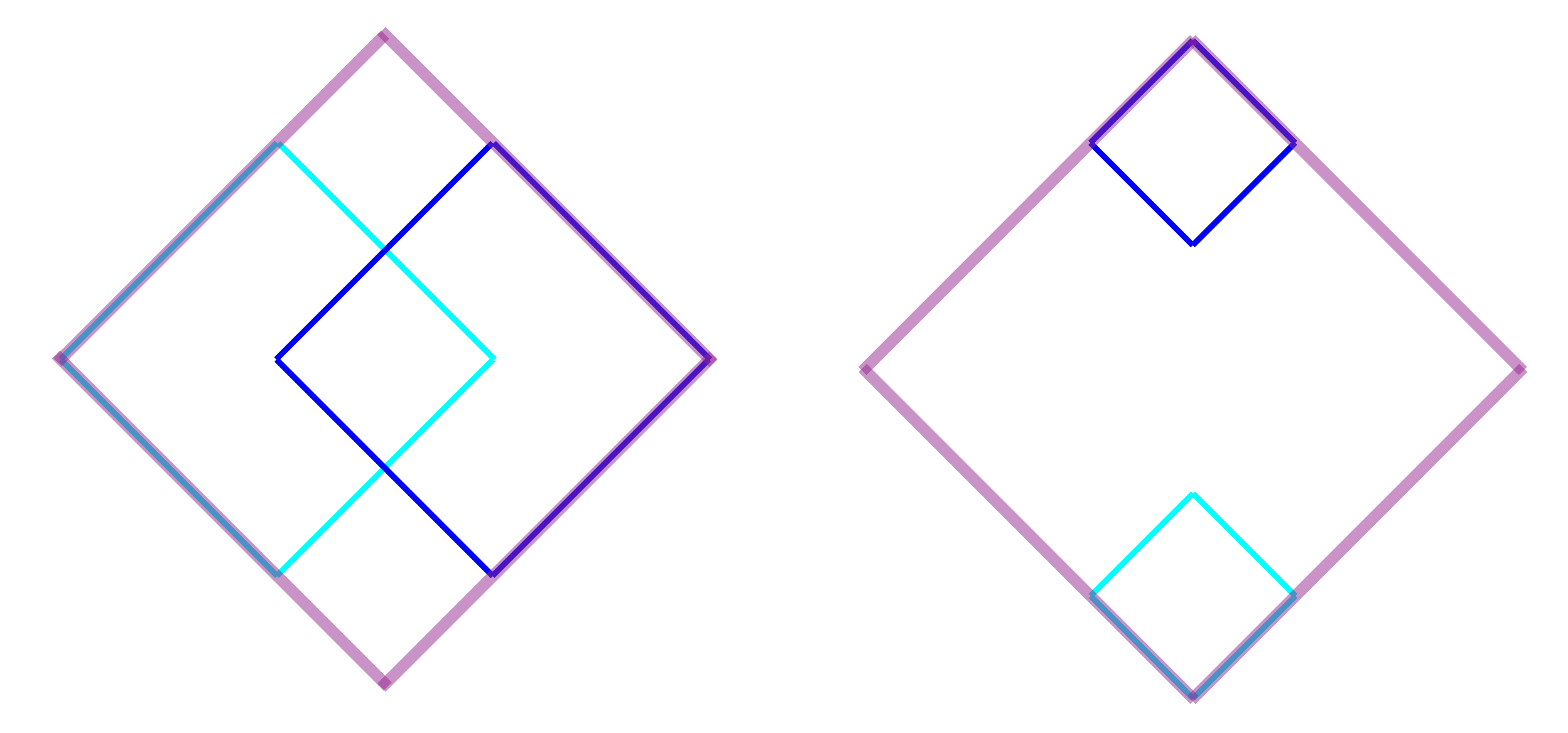} 
\end{center}
\caption{The wedge union (thick purple diamond) of two wedges (blue diamonds) is defined as the wedge complement of the intersection of their wedge complements. Two examples are shown.}
\label{fig:union}
\end{figure}
\begin{defn}
Given a wedge $a$ and a point $p\in \eth a$, the past (future) {\em quantum expansion}, $\Theta^-(a,p)$ ($\Theta^+(a,p)$), is the shape derivative of the generalized entropy under outward deformations of $a$ along the past (future) null vector field orthogonal to $\eth a$ at $p$. A precise definition can be given in terms of a functional derivative:
\begin{equation}
    \Theta^\pm(a,p)\equiv 4G\hbar \, \frac{\delta S_{\rm gen}[a(X^\pm(p))]}{\delta X^\pm(p)}~.
\end{equation}
Here $X^\pm$ are null coordinates orthogonal to $\eth a$, and $a(X^\pm(p))$ are wedges obtained by deforming $\eth a$ along them. See Ref.~\cite{Bousso:2015mna} for further details.
\end{defn}

\begin{rem}
By Eq.~\eqref{sgendef}, the quantum expansion can be decomposed into a classical and a quantum piece:
\begin{equation} \label{eq:qexpansion}
    \Theta^\pm(a,p) = \theta(p)+4G\hbar\, \frac{\delta S[a(X^\pm(p))]}{\delta X^\pm(p)}~.
\end{equation}
The first term is the classical expansion~\cite{Wald:1984rg}, which depends only on the shape of $\eth a$ near $p$. The second term in Eq.~\eqref{eq:qexpansion} is nonlocal.
\end{rem}

\begin{lem}\label{subsetexplem}
Let $a\subset b$ be wedges whose edges coincide in an open neighborhood $O$. At any point $p\in \eth a \cap O$, 
\begin{equation}%%%\label{subsetexp}
    \Theta^\pm(a,p)\geq \Theta^\pm(b,p)~.
\end{equation}
\end{lem}

\begin{proof}
Since $\eth a$ and $\eth b$ coincide near $p$, their shape, and hence their classical expansion $\theta$, will agree at $p$.
The Lemma then follows from strong subadditivity of the von Neumann entropy, applied to an infinitesimal deformation of $a$, and of $b$, at $p$; see Fig. 3a in Ref.~\cite{Bousso:2015mna}.
\end{proof}

\begin{defn}
The wedge $a$ is called {\em quantum-normal} at $p\in \eth a$ if $\Theta^+(a,p)\geq 0$ and $\Theta^-(a,p)\geq 0$. Other combinations of signs correspond to \emph{quantum-antinormal} ($\leq,\leq$), \emph{quantum-trapped} ($\leq,\geq$), \emph{quantum-antitrapped} ($\geq,\leq$) and \emph{quantum extremal} ($=,=$). Marginal cases arise if one expansion vanishes at $p$. In relations that hold for all $p\in \eth a$, we drop the argument $p$.
\end{defn}

\begin{conj}[Quantum Focussing Conjecture] \label{conj-qfc}
At all orders in $G\hbar$~\cite{Bousso:2015mna},
\begin{equation}\label{qfc}
   \frac{\delta \Theta(a,p)}{\delta X^\pm(\bar{p})}\leq 0~.
\end{equation}
\end{conj}
\begin{rem}
In the ``off-diagonal'' case, $p\neq \bar{p}$, the Quantum Focussing Conjecture (QFC) follows from Lemma~\ref{subsetexplem}. A quantum field theory limit of the QFC, the Quantum Null Energy Condition~\cite{Bousso:2015mna}, was proven in Refs.~\cite{Bousso:2015wca, Balakrishnan:2017bjg}; see also Refs.~\cite{Koeller:2015qmn, Wall:2017blw, Ceyhan:2018zfg, Balakrishnan:2019gxl}. The general diagonal case, $p=\bar{p}$, remains a conjecture.
\end{rem}

\begin{lem}\label{qnormalwi}
The intersection of two quantum-normal wedges is quantum-normal.
\end{lem}

\begin{proof}
Let $a$ and $b$ be quantum-normal. We use the decomposition of $\eth(a \cap b)$ described in Fig. \ref{fig-wedges}. For $p \in \eth a\cap b$ (or $p \in a\cap \eth b$), the condition $\Theta^\pm(a \cap b, p) \geq 0$ follows from the normalcy of $a$ (or $b$) at $p$ and Lemma~\ref{subsetexplem}.  For $p\in H^\pm (a)\cap H^\mp(b)$, let $a_p \subset a$ and $b_p \subset b$ be deformations of $a$ and $b$ along $H^\pm (a)$ and $H^\mp (b)$ respectively such that $\eth a_p, \eth b_p$ coincide with $H^\mp (a)\cap H^\pm(b)$ in a neighbourhood of $p$. Since $\Theta^\pm(a) \geq 0$ and  $\Theta^\mp(b) \geq 0$ in a neighbourhood of the points lightlike separated from $p$, Conjecture~\ref{conj-qfc} implies that $\Theta^\pm(a_p,p)\geq 0$ and  $\Theta^\mp(b_p,p) \geq 0$. Hence $\Theta^\pm(a \cap b, p) \geq 0$ by Lemma~\ref{subsetexplem}. 
\end{proof}

\subsection{Smallest-Generalized-Entropy Normal Wedge}
\label{sec-ssnw}

\begin{defn}\label{qew}
Given a wedge $a$, let $E_n(a)$ denote the quantum-normal wedge that contains $a$, has the same conformal boundary as $a$, and has the smallest generalized entropy among all such wedges. As in the static case (Def.~\ref{ewstatic}), we assume without proof that this wedge exists, and we assume for convenience that it is unique.
\end{defn}

\begin{lem}[Monotonicity]
Let $a$ and $b$ be wedges with the same conformal boundary, $\tilde\eth a=\tilde\eth b$. Then 
\begin{equation}
    a\subset b\implies \S[E_n(a)]\leq \S[E_n(b)]~.
\end{equation}
\end{lem}
\begin{proof}
The proof is analogous to the static case, Lemma \ref{monolem}.
\end{proof}

\begin{lem}[Quantum Expansion] \label{qeqe}
By definition, $E_n$ is quantum normal. Marginal cases arise at points $p\in\eth[E_n(a)]$ that do not lie on $\eth a$:
\begin{itemize}
    \item $E_n(a)$ is marginally quantum-antitrapped, $\Theta^-[E_n(a),p]=0$, at $p\in H^+(a')$.
    \item $E_n(a)$ is marginally quantum-trapped, $\Theta^+[E_n(a),p]=0$, at points $p\in H^-(a')$.
    \item $E_n(a)$ is quantum extremal, $\Theta^+[E_n(a),p]=\Theta^-[E_n(a);p]=0$, at points $p\in a'$.
\end{itemize} 
\end{lem}
\begin{proof}
By quantum normalcy of $E(a)$, $\S$ decreases when $\eth E(a)$ is deformed towards $\eth a$ at $p$. Lemma \ref{subsetexplem} ensures quantum normalcy of $\eth E(a)$ away from $p$ under this inward deformation. This conflicts with the $\S$-minimization requirement in Def.~\ref{qew}, unless the quantum expansion at $p$ in the other null direction immediately becomes negative under the deformation. 
\end{proof}

\begin{defn}[Entanglement wedge of a boundary region] \label{defn:bdyqew}
Let $B \subset \partial \tilde{M}$ be a subregion of a Cauchy surface of the conformal boundary  $\partial \tilde{M}$. The \emph{entanglement wedge of} of $B$, $\mathrm{EW}(B)\subset M$, is the smallest-generalized-entropy quantum-extremal wedge with conformal edge $B$~\cite{Engelhardt:2014gca}.
\end{defn}

\begin{lem}[Reduction of $E_n$ to $\mathrm{EW}$]\label{lem:qgenew=ew}
Let $a$ be contained in the ordinary quantum entanglement wedge of its conformal edge $\tilde\eth a$, \emph{i.e.}, $a \subset \mathrm{EW}(\tilde\eth a)$. Then
\begin{align}
E_n(a) = \mathrm{EW}(\tilde\eth a).
\end{align}
\end{lem}
\begin{proof}
By the arguments in the proof of Lemma \ref{qeqe}, the smallest-$\S$ quantum-normal wedge with conformal edge $\tilde\eth a$ is necessarily extremal, and hence is given by $\mathrm{EW}(\tilde\eth a)$. By assumption $a \subset \mathrm{EW}(\tilde\eth a)$, so this continues to be true when we restrict to normal wedges containing $a$.
\end{proof}

\begin{lem}[Reduction of $E_n$ to $E$] \label{lem:entoe}
Let $\Sigma_0$ be a time-reflection symmetric Cauchy surface, and let $a$ be a wedge with $\eth a\subset \Sigma_0$. Then $E_n(a)$ reduces to the static entanglement wedge $E(a)$:\footnote{In Sec.~\ref{sec-prescription}, $\partial$ denotes the boundary in the topology of $\Sigma_0$, and we revert to this usage in this lemma and its proof. Elsewhere in the present section, $\partial$ denotes the boundary in $M$.}
\begin{equation}
    \eth E_n(a)= \partial E(a)~.
\end{equation}
\end{lem}
\begin{proof}
Let $T$ be the time-reflection operator.
The set $E(a) \cap T[E(a)]$ is a wedge by Remark~\ref{wilem}, normal by Lemma \ref{qnormalwi}, and it contains $a$. It is also time-reflection symmetric, so its edge lies on $\Sigma_0$. This conflicts with Conj.~\ref{conj-qfc} unless $E(a)=T[E(a)]$. Hence $\eth E_n(a) \subset \Sigma_0$. Moreover, $\partial E(a\cap \Sigma_0)$ must be pointwise quantum-normal or antinormal by time-reflection symmetry. If it were anti-normal at some point $p$, we could decrease $\S$ by deforming $\partial E(a\cap\Sigma_0)$ outward on $\Sigma_0$ at $p$, contradicting its definition.
\end{proof}

\begin{thm}[No Cloning for $E_n$]%%%\label{qnc}
Let $a,b$ be wedges that satisfy
\begin{equation}\label{qexclusions}
    a\subset E_n(b)'~, ~~ b\subset E_n(a)'~.
\end{equation}
Then
\begin{equation}\label{qnocloneq}
    E_n(a)\subset E_n(b)'~,
\end{equation}
or equivalently, $E_n(b)\subset E_n(a)'$.
\end{thm}
\begin{proof}
For arbitrary wedges $c$ and $d$, let
\begin{equation}\label{wp} 
    c_{/d}\equiv (c\cap d)\Cup(c\cap d')~.
\end{equation}
Lemma~\ref{qeqe} and Conj.~\ref{conj-qfc} can be shown to imply
\begin{equation}
    \S[E_n(a)'_{/E_n(b)}]\leq \S[E_n(a)]~,\,\,\,\mathrm{and}\,\,\,\S[E_n(b)'_{/E_n(a)}]\leq \S[E_n(b)]~.
\end{equation}
We add these inequalities and use strong subadditivity of the von Neumann entropy to obtain
\begin{equation}\label{qseaebp}
   \S[E_n(a)\cap E_n(b)'] + \S[E_n(a)'\cap E_n(b)] \leq  \S[E_n(a)]+\S[E_n(b)]~.
\end{equation}
We will now show that $E_n(a)\cap E_n(b)'$ is quantum-normal. Let us begin by substituting $a\to E_n(a)$ and $b\to E_n(b)'$ in the proof of Lemma~\ref{qnormalwi}. The Lemma assumed that $a$ and $b$ are normal, whereas now $E_n(b)'$ is quantum-antinormal. However, our assumption, Eq.~\eqref{qexclusions}, implies that $\Theta^+[E_n(b)',p]=0$ or $\Theta^-[E_n(b)',p]=0$ at any points $p$ where these expansions actually appear in the proof. Indeed, since $E_n(a)$ is spacelike to $b$, $\Theta^\pm[E_n(b)',p]=0$ for $p\in E_n(a)\cap \eth E_n(b)$ by Lemma~\ref{qeqe}. If $p$ is the endpoint of a geodesic generator of $H^-[E_n(b)']$ that intersects $H^+[E_n(a)]$, then $\Theta^-[E_n(b)',p]=0$ by Lemma~\ref{qeqe}, or else there would be a causal curve from $\eth E_n(a)$ to $b$, in conflict with Eq.~\eqref{qexclusions}. 

By the same reasoning $E_n(a)'\cap E_n(b)$ is also quantum-normal. Moreover, $a\subset E_n(a)\cap E_n(b)'$ and $b\subset E_n(a)'\cap E_n(b)$ by Eq.~\eqref{qexclusions}. By Def.~\ref{qew}, $E_n(a)$ and $E_n(b)$ are the unique smallest-$\S$ quantum-normal wedges containing, respectively, $a$ and $b$. Eq.~\eqref{qseaebp} thus implies $E_n(a)\cap E_n(b)'=E_n(a)$, which is equivalent to Eq.~\eqref{qnocloneq}.
\end{proof}

\subsection{Discussion}
\label{sec-covariantdiscussion}

$E_n$ does not satisfy subadditivity or nesting. This suggests that $E_n$ is not quite the correct generalization of $E$. A counterexample to both properties is shown in Fig.~\ref{ShellCE}. A purely classical counterexample to strong subadditivity is furnished by taking $b$ to be a rectangle of length and width $\Delta y\gg \tau \gg \Delta x\gg \ell_P$ at $t=0$ in two-dimensional Minkowski space, and $a$ and $c$ to be identical rectangles obtained by moving $b$ by $1+2\Delta x$ in the $\pm x$ direction and by $\tau$ in the $t$ direction.
\begin{figure}[t]
\begin{center}
  \includegraphics[width = 0.6\linewidth]{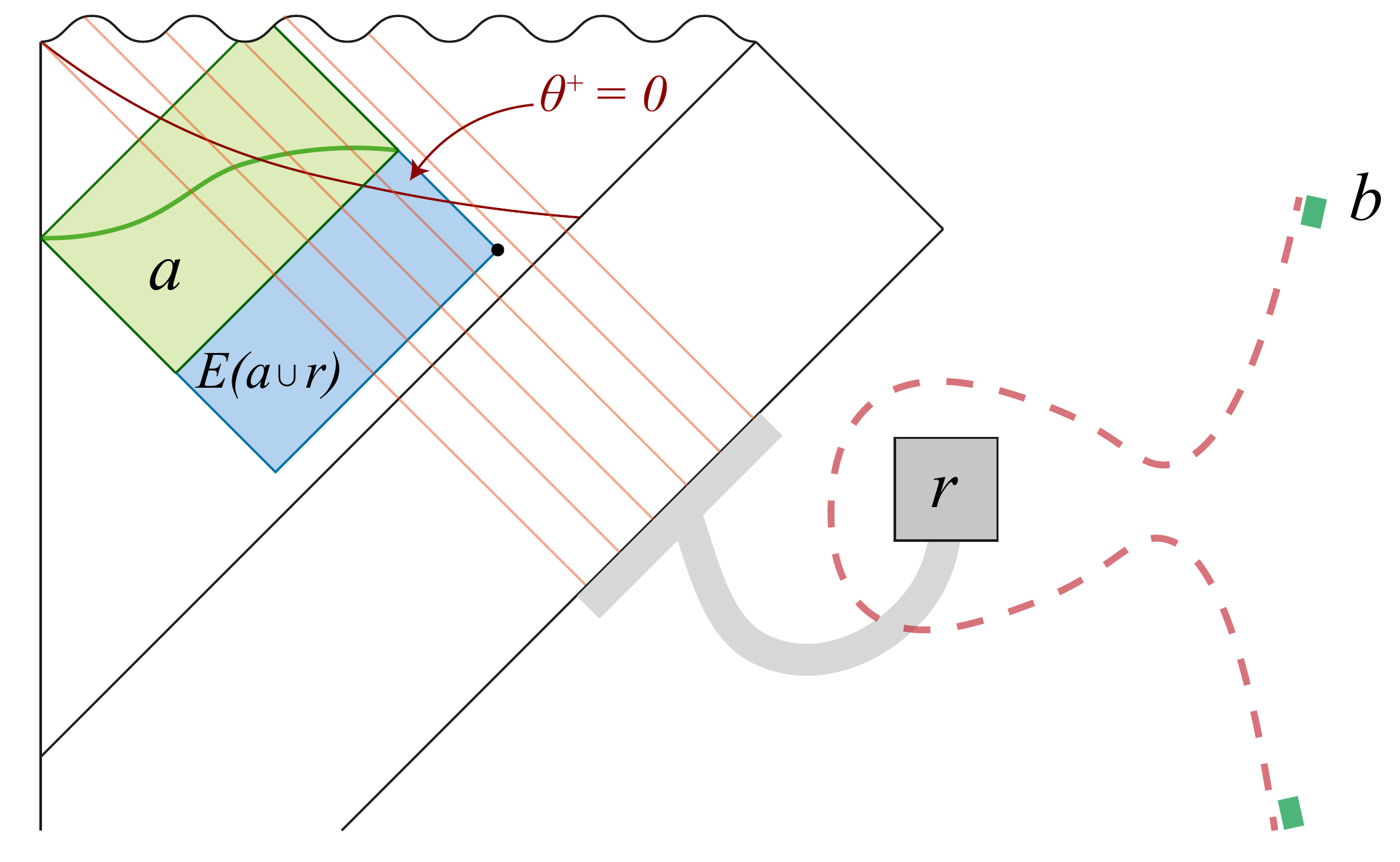} 
\end{center}
\caption{Collapse of a null shell purified by a distant reference system $r$. The wedge $a$ is marginally quantum trapped, so $a=E_n(a)$. One can arrange that $\S[E_n(a)]+\S[E_n (r)]\gg \S[E_n(a\cup r)]$. If $r$ is the Python's Lunch~\cite{Brown:2019rox} (dashed red) of an asymptotic region $b$, one can arrange that $E_n(a\cup b)\not\supset E_n(b)$. Hence $E_n$ satisfies neither subadditivity nor nesting.}
\label{ShellCE}
\end{figure}

Ideally, we would like to find a prescription $e(a)$ that satisfies all of the following properties:
\begin{enumerate}
    \item $e$ reduces to the time-reflection symmetric prescription given in Sec.~\ref{sec-prescription}: $e(a) = E(a)$.
    \item $e$ reduces to the QES prescription for the entanglement wedge of boundary regions: $a\subset EW(\deltabar a)\implies e(a) = EW(\deltabar a)$.
    \item $e$ is well-defined.
    \item $e$ satisfies no-cloning, strong subadditivity, and nesting, in the sense shown for $E$ in Sec.~\ref{sec-prescription}.
\end{enumerate}

We have seen that $E_n$ satisfies the first three requirements as well as no-cloning. We are aware of alternate candidate prescriptions satisfying other subsets of these requirements, but of no \emph{single} prescription that satisfies all. For example, one can ensure nesting in addition to the requirements satisfied by $E_n$, by taking the wedge union $\Cup_{b \subset a} E_n(b)$ over all subwedges $b \subset a$. However, this definition is somewhat artificial, and it still does not obey strong subadditivity.

Ref.~\cite{Grado-White:2020wlb} proposed a restricted maximin prescription for the entanglement wedge of a cut-off boundary region. This prescription can be converted into a proposal for the generalized entanglement wedge of a bulk region, which (like $E_n$) satisfies the first three requirements above. In addition, it obeys strong subadditivity; however, it violates nesting.\footnote{We also expect no-cloning to hold but we have not verified this.} Explicitly, let $\eth a$ be the cut-off surface denoted $\gamma$ in Ref.~\cite{Grado-White:2020wlb}, and let $a$ be its exterior, \emph{i.e.}, the portion discarded by the cut-off.
%if one replaces the cut-off spacetime by $a'$ and the surface $\gamma$ by the edge $\eth a$, then
Choosing $A=\gamma$ in their notation, the prescription of Ref.~\cite{Grado-White:2020wlb} defines a wedge $\mathcal{E}(a) \subset a'$ bounded by $\eth a$ and its restricted maximin surface. The wedge union $e(a) = a \,\Cup \,\mathcal{E}(a)$ is then a candidate for the generalized entanglement wedge. The proof of strong subadditivity from Ref.~\cite{Grado-White:2020wlb} directly implies strong subadditivity for the generalized entanglement wedges, so long as $e(a \Cup b) \cap c = e(b \Cup c) \cap a = \varnothing$. (This ensures that the maximin surfaces lie in $a'\cap b' \cap c'$, the spacetime region that survives all three cutoffs. From an information-theoretic perspective, it is roughly analogous to the conditions in Theorem \ref{staticnocloning} which ensure that $a$, $b$ and $c$ are genuinely independent degrees of freedom.) %Moreover, the first three requirements above are easily shown to be satisfied. 
However, the statement of nesting proved in Ref.~\cite{Grado-White:2020wlb}, while appropriate in that context, is very different from the general statement $e(a) \subset e(b)$ that must hold for $e$ to be an entanglement wedge of a bulk region. In fact, it is easy to see that the latter statement fails, since $e(a) = a$ for any region $a$ with edge $\eth a$ piecewise lightlike.

However, it may not be appropriate to seek a single definition of the entanglement wedge of bulk regions. In recent work, it was shown that the QES prescription for the entanglement wedge $EW$ of boundary regions must be refined, by replacing the von Neumann entropy with certain one-shot entropies~\cite{Akers:2020pmf}. For static spacetimes, instead of a single entanglement wedge, the refined prescription defines \emph{two} wedges, called the max-EW and the min-EW. They describe the bulk regions that are respectively fully encoded in, and partially influenced by, the boundary region. This refinement has been generalized to the full time-dependent setting~\cite{Akers:2022toappear}. For a boundary input region, the difference between its max- and min-EW is always a quantum effect. However, if one tries to generalize this prescription further to allow for bulk input regions, it appears that in time-dependent settings, the max- and min-EW may differ even at the classical level. We are hopeful that this may resolve the difficulties that we encountered in this section when trying to define a single generalized entanglement wedge with all our desired properties. We leave the details of any such prescription to future work.

\subsection*{Acknowledgements}
We thank Chris Akers, Ben Freivogel, Adam Levine, and Juan Maldacena for valuable discussions. This work was supported in part by the Berkeley Center for Theoretical Physics; by the Department of Energy, Office of Science, Office of High Energy Physics under QuantISED Award DE-SC0019380 and under contract DE-AC02-05CH11231. RB was supported by the National Science Foundation under Award Number 2112880. GP was supported by the Simons Foundation through the ``It from Qubit'' program; by AFOSR award FA9550-22-1-0098; DOE award DE-FOA-0002563; and by an IBM Einstein Fellowship at the Institute for Advanced Study.

\bibliographystyle{JHEP}
\bibliography{bibliography}

\providecommand{\href}[2]{#2}\begingroup\raggedright\begin{thebibliography}{10}

\bibitem{Maldacena:1997re}
J.~M. Maldacena, \emph{{The Large N limit of superconformal field theories and
  supergravity}}, \href{https://doi.org/10.1023/A:1026654312961}{\emph{Adv.
  Theor. Math. Phys.} {\bfseries 2} (1998) 231}
  [\href{https://arxiv.org/abs/hep-th/9711200}{{\ttfamily hep-th/9711200}}].

\bibitem{Ryu:2006bv}
S.~Ryu and T.~Takayanagi, \emph{{Holographic derivation of entanglement entropy
  from AdS/CFT}},
  \href{https://doi.org/10.1103/PhysRevLett.96.181602}{\emph{Phys. Rev. Lett.}
  {\bfseries 96} (2006) 181602}
  [\href{https://arxiv.org/abs/hep-th/0603001}{{\ttfamily hep-th/0603001}}].

\bibitem{Hubeny:2007xt}
V.~E. Hubeny, M.~Rangamani and T.~Takayanagi, \emph{{A Covariant holographic
  entanglement entropy proposal}},
  \href{https://doi.org/10.1088/1126-6708/2007/07/062}{\emph{JHEP} {\bfseries
  07} (2007) 062} [\href{https://arxiv.org/abs/0705.0016}{{\ttfamily
  0705.0016}}].

\bibitem{Faulkner:2013ana}
T.~Faulkner, A.~Lewkowycz and J.~Maldacena, \emph{{Quantum corrections to
  holographic entanglement entropy}},
  \href{https://doi.org/10.1007/JHEP11(2013)074}{\emph{JHEP} {\bfseries 11}
  (2013) 074} [\href{https://arxiv.org/abs/1307.2892}{{\ttfamily 1307.2892}}].

\bibitem{Engelhardt:2014gca}
N.~Engelhardt and A.~C. Wall, \emph{{Quantum Extremal Surfaces: Holographic
  Entanglement Entropy beyond the Classical Regime}},
  \href{https://doi.org/10.1007/JHEP01(2015)073}{\emph{JHEP} {\bfseries 01}
  (2015) 073} [\href{https://arxiv.org/abs/1408.3203}{{\ttfamily 1408.3203}}].

\bibitem{Hayden:2018khn}
P.~Hayden and G.~Penington, \emph{{Learning the Alpha-bits of Black Holes}},
  \href{https://doi.org/10.1007/JHEP12(2019)007}{\emph{JHEP} {\bfseries 12}
  (2019) 007} [\href{https://arxiv.org/abs/1807.06041}{{\ttfamily
  1807.06041}}].

\bibitem{Akers:2020pmf}
C.~Akers and G.~Penington, \emph{{Leading order corrections to the quantum
  extremal surface prescription}},
  \href{https://doi.org/10.1007/JHEP04(2021)062}{\emph{JHEP} {\bfseries 04}
  (2021) 062} [\href{https://arxiv.org/abs/2008.03319}{{\ttfamily
  2008.03319}}].

\bibitem{Wall:2012uf}
A.~C. Wall, \emph{{Maximin Surfaces, and the Strong Subadditivity of the
  Covariant Holographic Entanglement Entropy}},
  \href{https://doi.org/10.1088/0264-9381/31/22/225007}{\emph{Class. Quant.
  Grav.} {\bfseries 31} (2014) 225007}
  [\href{https://arxiv.org/abs/1211.3494}{{\ttfamily 1211.3494}}].

\bibitem{Jafferis:2015del}
D.~L. Jafferis, A.~Lewkowycz, J.~Maldacena and S.~J. Suh, \emph{{Relative
  entropy equals bulk relative entropy}},
  \href{https://doi.org/10.1007/JHEP06(2016)004}{\emph{JHEP} {\bfseries 06}
  (2016) 004} [\href{https://arxiv.org/abs/1512.06431}{{\ttfamily
  1512.06431}}].

\bibitem{Bousso:2012sj}
R.~Bousso, S.~Leichenauer and V.~Rosenhaus, \emph{{Light-sheets and AdS/CFT}},
  \href{https://doi.org/10.1103/PhysRevD.86.046009}{\emph{Phys. Rev. D}
  {\bfseries 86} (2012) 046009}
  [\href{https://arxiv.org/abs/1203.6619}{{\ttfamily 1203.6619}}].

\bibitem{Hubeny:2012wa}
V.~E. Hubeny and M.~Rangamani, \emph{{Causal Holographic Information}},
  \href{https://doi.org/10.1007/JHEP06(2012)114}{\emph{JHEP} {\bfseries 06}
  (2012) 114} [\href{https://arxiv.org/abs/1204.1698}{{\ttfamily 1204.1698}}].

\bibitem{Czech:2012bh}
B.~Czech, J.~L. Karczmarek, F.~Nogueira and M.~Van~Raamsdonk, \emph{{The
  Gravity Dual of a Density Matrix}},
  \href{https://doi.org/10.1088/0264-9381/29/15/155009}{\emph{Class. Quant.
  Grav.} {\bfseries 29} (2012) 155009}
  [\href{https://arxiv.org/abs/1204.1330}{{\ttfamily 1204.1330}}].

\bibitem{Penington:2019npb}
G.~Penington, \emph{{Entanglement Wedge Reconstruction and the Information
  Paradox}}, \href{https://doi.org/10.1007/JHEP09(2020)002}{\emph{JHEP}
  {\bfseries 09} (2020) 002}
  [\href{https://arxiv.org/abs/1905.08255}{{\ttfamily 1905.08255}}].

\bibitem{Almheiri:2019psf}
A.~Almheiri, N.~Engelhardt, D.~Marolf and H.~Maxfield, \emph{{The entropy of
  bulk quantum fields and the entanglement wedge of an evaporating black
  hole}}, \href{https://doi.org/10.1007/JHEP12(2019)063}{\emph{JHEP} {\bfseries
  12} (2019) 063} [\href{https://arxiv.org/abs/1905.08762}{{\ttfamily
  1905.08762}}].

\bibitem{Bousso:2019ykv}
R.~Bousso and M.~Toma\v{s}evi\'c, \emph{{Unitarity From a Smooth Horizon?}},
  \href{https://doi.org/10.1103/PhysRevD.102.106019}{\emph{Phys. Rev. D}
  {\bfseries 102} (2020) 106019}
  [\href{https://arxiv.org/abs/1911.06305}{{\ttfamily 1911.06305}}].

\bibitem{Lewkowycz:2013nqa}
A.~Lewkowycz and J.~Maldacena, \emph{{Generalized gravitational entropy}},
  \href{https://doi.org/10.1007/JHEP08(2013)090}{\emph{JHEP} {\bfseries 08}
  (2013) 090} [\href{https://arxiv.org/abs/1304.4926}{{\ttfamily 1304.4926}}].

\bibitem{Akers:2019wxj}
C.~Akers, S.~Leichenauer and A.~Levine, \emph{{Large Breakdowns of Entanglement
  Wedge Reconstruction}},
  \href{https://doi.org/10.1103/PhysRevD.100.126006}{\emph{Phys. Rev. D}
  {\bfseries 100} (2019) 126006}
  [\href{https://arxiv.org/abs/1908.03975}{{\ttfamily 1908.03975}}].

\bibitem{Bekenstein:1972tm}
J.~D. Bekenstein, \emph{{Black holes and the second law}},
  \href{https://doi.org/10.1007/BF02757029}{\emph{Lett. Nuovo Cim.} {\bfseries
  4} (1972) 737}.

\bibitem{Wall:2011hj}
A.~C. Wall, \emph{{A proof of the generalized second law for rapidly changing
  fields and arbitrary horizon slices}},
  \href{https://doi.org/10.1103/PhysRevD.85.104049}{\emph{Phys. Rev. D}
  {\bfseries 85} (2012) 104049}
  [\href{https://arxiv.org/abs/1105.3445}{{\ttfamily 1105.3445}}].

\bibitem{Bousso:1999xy}
R.~Bousso, \emph{{A Covariant entropy conjecture}},
  \href{https://doi.org/10.1088/1126-6708/1999/07/004}{\emph{JHEP} {\bfseries
  07} (1999) 004} [\href{https://arxiv.org/abs/hep-th/9905177}{{\ttfamily
  hep-th/9905177}}].

\bibitem{Bousso:2015mna}
R.~Bousso, Z.~Fisher, S.~Leichenauer and A.~C. Wall, \emph{{Quantum focusing
  conjecture}}, \href{https://doi.org/10.1103/PhysRevD.93.064044}{\emph{Phys.
  Rev. D} {\bfseries 93} (2016) 064044}
  [\href{https://arxiv.org/abs/1506.02669}{{\ttfamily 1506.02669}}].

\bibitem{Dong:2020uxp}
X.~Dong, X.-L. Qi, Z.~Shangnan and Z.~Yang, \emph{{Effective entropy of quantum
  fields coupled with gravity}},
  \href{https://doi.org/10.1007/JHEP10(2020)052}{\emph{JHEP} {\bfseries 10}
  (2020) 052} [\href{https://arxiv.org/abs/2007.02987}{{\ttfamily
  2007.02987}}].

\bibitem{Grado-White:2020wlb}
B.~Grado-White, D.~Marolf and S.~J. Weinberg, \emph{{Radial Cutoffs and
  Holographic Entanglement}},
  \href{https://doi.org/10.1007/JHEP01(2021)009}{\emph{JHEP} {\bfseries 01}
  (2021) 009} [\href{https://arxiv.org/abs/2008.07022}{{\ttfamily
  2008.07022}}].

\bibitem{Raju:2021lwh}
S.~Raju, \emph{{Failure of the split property in gravity and the information
  paradox}},  \href{https://arxiv.org/abs/2110.05470}{{\ttfamily 2110.05470}}.

\bibitem{Banks:1998dd}
T.~Banks, M.~R. Douglas, G.~T. Horowitz and E.~J. Martinec, \emph{{AdS dynamics
  from conformal field theory}},
  \href{https://arxiv.org/abs/hep-th/9808016}{{\ttfamily hep-th/9808016}}.

\bibitem{Hamilton:2006az}
A.~Hamilton, D.~N. Kabat, G.~Lifschytz and D.~A. Lowe, \emph{{Holographic
  representation of local bulk operators}},
  \href{https://doi.org/10.1103/PhysRevD.74.066009}{\emph{Phys. Rev. D}
  {\bfseries 74} (2006) 066009}
  [\href{https://arxiv.org/abs/hep-th/0606141}{{\ttfamily hep-th/0606141}}].

\bibitem{Harlow:2018tng}
D.~Harlow and H.~Ooguri, \emph{{Symmetries in quantum field theory and quantum
  gravity}}, \href{https://doi.org/10.1007/s00220-021-04040-y}{\emph{Commun.
  Math. Phys.} {\bfseries 383} (2021) 1669}
  [\href{https://arxiv.org/abs/1810.05338}{{\ttfamily 1810.05338}}].

\bibitem{Harlow:2013tf}
D.~Harlow and P.~Hayden, \emph{{Quantum Computation vs. Firewalls}},
  \href{https://doi.org/10.1007/JHEP06(2013)085}{\emph{JHEP} {\bfseries 06}
  (2013) 085} [\href{https://arxiv.org/abs/1301.4504}{{\ttfamily 1301.4504}}].

\bibitem{Brown:2019rox}
A.~R. Brown, H.~Gharibyan, G.~Penington and L.~Susskind, \emph{{The
  Python\textquoteright{}s Lunch: geometric obstructions to decoding Hawking
  radiation}}, \href{https://doi.org/10.1007/JHEP08(2020)121}{\emph{JHEP}
  {\bfseries 08} (2020) 121}
  [\href{https://arxiv.org/abs/1912.00228}{{\ttfamily 1912.00228}}].

\bibitem{Swingle:2009bg}
B.~Swingle, \emph{{Entanglement Renormalization and Holography}},
  \href{https://doi.org/10.1103/PhysRevD.86.065007}{\emph{Phys. Rev. D}
  {\bfseries 86} (2012) 065007}
  [\href{https://arxiv.org/abs/0905.1317}{{\ttfamily 0905.1317}}].

\bibitem{Pastawski:2015qua}
F.~Pastawski, B.~Yoshida, D.~Harlow and J.~Preskill, \emph{{Holographic quantum
  error-correcting codes: Toy models for the bulk/boundary correspondence}},
  \href{https://doi.org/10.1007/JHEP06(2015)149}{\emph{JHEP} {\bfseries 06}
  (2015) 149} [\href{https://arxiv.org/abs/1503.06237}{{\ttfamily
  1503.06237}}].

\bibitem{Hayden:2016cfa}
P.~Hayden, S.~Nezami, X.-L. Qi, N.~Thomas, M.~Walter and Z.~Yang,
  \emph{{Holographic duality from random tensor networks}},
  \href{https://doi.org/10.1007/JHEP11(2016)009}{\emph{JHEP} {\bfseries 11}
  (2016) 009} [\href{https://arxiv.org/abs/1601.01694}{{\ttfamily
  1601.01694}}].

\bibitem{Bao:2018pvs}
N.~Bao, G.~Penington, J.~Sorce and A.~C. Wall, \emph{{Beyond Toy Models:
  Distilling Tensor Networks in Full AdS/CFT}},
  \href{https://doi.org/10.1007/JHEP11(2019)069}{\emph{JHEP} {\bfseries 11}
  (2019) 069} [\href{https://arxiv.org/abs/1812.01171}{{\ttfamily
  1812.01171}}].

\bibitem{Wald:1984rg}
R.~M. Wald, \emph{{General Relativity}}. Chicago Univ. Pr., Chicago, USA, 1984,
  \href{https://doi.org/10.7208/chicago/9780226870373.001.0001}{10.7208/chicago/9780226870373.001.0001}.

\bibitem{Bousso:2015wca}
R.~Bousso, Z.~Fisher, J.~Koeller, S.~Leichenauer and A.~C. Wall, \emph{{Proof
  of the Quantum Null Energy Condition}},
  \href{https://doi.org/10.1103/PhysRevD.93.024017}{\emph{Phys. Rev. D}
  {\bfseries 93} (2016) 024017}
  [\href{https://arxiv.org/abs/1509.02542}{{\ttfamily 1509.02542}}].

\bibitem{Balakrishnan:2017bjg}
S.~Balakrishnan, T.~Faulkner, Z.~U. Khandker and H.~Wang, \emph{{A General
  Proof of the Quantum Null Energy Condition}},
  \href{https://doi.org/10.1007/JHEP09(2019)020}{\emph{JHEP} {\bfseries 09}
  (2019) 020} [\href{https://arxiv.org/abs/1706.09432}{{\ttfamily
  1706.09432}}].

\bibitem{Koeller:2015qmn}
J.~Koeller and S.~Leichenauer, \emph{{Holographic Proof of the Quantum Null
  Energy Condition}},
  \href{https://doi.org/10.1103/PhysRevD.94.024026}{\emph{Phys. Rev. D}
  {\bfseries 94} (2016) 024026}
  [\href{https://arxiv.org/abs/1512.06109}{{\ttfamily 1512.06109}}].

\bibitem{Wall:2017blw}
A.~C. Wall, \emph{{Lower Bound on the Energy Density in Classical and Quantum
  Field Theories}},
  \href{https://doi.org/10.1103/PhysRevLett.118.151601}{\emph{Phys. Rev. Lett.}
  {\bfseries 118} (2017) 151601}
  [\href{https://arxiv.org/abs/1701.03196}{{\ttfamily 1701.03196}}].

\bibitem{Ceyhan:2018zfg}
F.~Ceyhan and T.~Faulkner, \emph{{Recovering the QNEC from the ANEC}},
  \href{https://doi.org/10.1007/s00220-020-03751-y}{\emph{Commun. Math. Phys.}
  {\bfseries 377} (2020) 999}
  [\href{https://arxiv.org/abs/1812.04683}{{\ttfamily 1812.04683}}].

\bibitem{Balakrishnan:2019gxl}
S.~Balakrishnan, V.~Chandrasekaran, T.~Faulkner, A.~Levine and
  A.~Shahbazi-Moghaddam, \emph{{Entropy Variations and Light Ray Operators from
  Replica Defects}},  \href{https://arxiv.org/abs/1906.08274}{{\ttfamily
  1906.08274}}.

\bibitem{Akers:2022toappear}
C.~Akers, A.~Levine, G.~Penington and E.~Wildenhain, \emph{{One-shot
  holography}},  \href{https://arxiv.org/abs/2307.13032}{{\ttfamily
  2307.13032}}.

\end{thebibliography}\endgroup

\end{document}